%% file: 2024physscr.tex
\documentclass[11pt]{article}
\input{qudit_preamble.tex}

\begin{document}

%==============================================================================================%
\title{\bf\Large Uncertainty relation and the constrained quadratic programming}
%==============================================================================================%

\author{\blue{Lin Zhang}$^1$\footnote{E-mail: godyalin@163.com},\quad \blue{Dade Wu}$^1$,\quad \blue{Ming-Jing Zhao}$^2$,\quad \blue{Hua Nan}$^3$\\
  {\it\small $^1$School of Science, Hangzhou Dianzi University, Hangzhou 310018, PR~China}\\
  {\it\small $^2$ School of Science, Beijing Information Science and Technology University, Beijing, 100192, PR~China}\\
  {\it\small $^3$Department of Mathematics, College of Sciences, Yanbian University, Yanji 133002, PR~China}}
\date{}
\maketitle

\begin{abstract}
The uncertainty relation is a fundamental concept in quantum theory,
plays a pivotal role in various quantum information processing
tasks. In this study, we explore the additive uncertainty relation
pertaining to two or more observables, in terms of their variance,
by utilizing the generalized Gell-Mann representation in qudit
systems. We find that the tight state-independent lower bound of the
variance sum can be characterized as a quadratic programming problem
with nonlinear constraints in optimization theory. As illustrative
examples, we derive analytical solutions for these quadratic
programming problems in lower-dimensional systems, which align with
the state-independent lower bounds. Additionally, we introduce a
numerical algorithm tailored for solving these quadratic programming
instances, highlighting its efficiency and accuracy. The advantage
of our approach lies in its potential ability to simultaneously
achieve the optimal value of the quadratic programming problem with
nonlinear constraints but also precisely identify the extremal state
where this optimal value is attained. This enables us to establish a
tight state-independent lower bound for the sum of variances, and
further identify the extremal state at which this lower bound is
realized.
\end{abstract}
\newpage
\tableofcontents
\newpage

%===================================================%
\section{Introduction}
%===================================================%

Quantum theory serves as a fundamental framework for elucidating the
behavior of matter and energy at atomic and subatomic scales.
Central to this theory is the uncertainty relation, a principle that
asserts the impossibility of simultaneously determining certain
physical properties of a particle with absolute precision. The
earliest formulation of the uncertainty relations stemmed from
Heisenberg's pioneering work \cite{Heisenberg1927}. Specifically,
for quantum systems characterized by a definite position and
momentum or a specific spin direction, it is inherently impossible
to definitively predict measurement outcomes, as quantum theory
instead provides probability distributions. Subsequently, Robertson
\cite{Robertson1929} broadened Heisenberg's uncertainty relation to
encompass any two observables within any finite-dimensional space.

Furthermore, the uncertainty relation imposes fundamental limits on
the extractable information from a given quantum system.
Measurements pertaining to one property often perturb or alter other
properties, thus constraining the amount of knowledge that can be
gained. Consequently, the uncertainty relations have been formulated
in various forms \cite{Friedland2013,Pucha2013}, including Shannon
entropy \cite{Deutsch1983,Wu2009,Coles2017}, R\'{e}nyi entropy
\cite{Maassen1988}, conditional entropy
\cite{Renes2009,Gour2018,Kurzyk2018}, and mutual information
\cite{Grudka2013}. These formulations have profound implications for
quantum information processing tasks, such as quantum key
distribution and two-party quantum cryptography
\cite{Coles2017,Tomamichel2011}. Additionally, the uncertainty
relation has emerged as a pivotal tool in quantum random number
generation \cite{Vallone2014,Cao2016}, entanglement witnessing
\cite{Berta2014}, EPR steering \cite{Walborn2011,Schneeloch2013},
and quantum metrology \cite{Giovannetti2011}.

Recently, the additive uncertainty relation has garnered significant
interest among researchers due to its profound insights into the
limitations of simultaneously measuring two observables
\cite{Zhang2021,Szymanski2019}. This relation has been extensively
studied both theoretically and experimentally
\cite{Schwonnek2017,Zhao2019}, posing a crucial challenge to our
understanding of quantum system behavior, particularly in the realm
of quantum technology such as quantum communication, cryptography,
and computation. Investigations in quantum information have
demonstrated novel applications of additive uncertainty, enabling
secure and efficient quantum information processing. This includes
the generation and manipulation of entangled states in quantum
systems, as well as the development of quantum error correction
codes. Overall, the additive uncertainty relation occupies a pivotal
position in quantum mechanics research, with implications extending
to both fundamental inquiries and practical quantum technology
advancements. Specifically, elucidating the interplay between
quantum uncertainty and entanglement
\cite{Maccone2014,Guhne2004,Akbari2018} could reveal novel
strategies for harnessing the unique properties of quantum systems
in information processing and communication tasks.

In this paper, motivated by the intricate relationship between the
sum uncertainty of two or more observables and the concept of
quantum entanglement, we aim to explore the additive uncertainty
relation of any two or more observables in terms of variance
\cite{Li2015,Qian2018}. To achieve this, we employ the generalized
Gell-Mann representation in qudit systems. Through our analysis, we
will reveal that the state-independent lower bound \cite{Xiao2019}
of the additive uncertainty relation in qudit systems is akin to a
quadratic programming problem involving nonlinear constraints in
optimization theory. Furthermore, we derive analytical lower bounds
for the additive uncertainty relation, which provide a deeper
understanding of its fundamental properties and potential
applications.

The remainder of this paper is structured as follows. In
Section~\ref{sect:2}, we present the preliminaries necessary for our
investigation. Subsequently, in Section~\ref{sect:3}, we present our
main findings, specifically the additive uncertainty relation of any
two or more observables in terms of variance, formulated as a
quadratic programming problem with nonlinear constraints
(Theorem~\ref{th:main}) in qudit systems. In Section~\ref{sect:4}
and Section~\ref{sect:5}, we focus on lower-dimensional cases,
exploring specific calculations for qubit and qutrit systems,
respectively, and characterizing the conditions for equality in
these cases through the explicit construction of extremal states.
Following these illustrative examples, Section~\ref{sect:6}
discusses the potential applications of the additive uncertainty
relation in entanglement detection. Finally, in
Section~\ref{sect:7}, we summarize our findings and conclude the
paper.

%===================================================%
\section{Preliminaries}\label{sect:2}
%===================================================%

Throughout this notes, all inner products involved in vectors are
Euclidean one, i.e., $\Inner{\bsx}{\bsy} := \sum_k\bar x_k y_k$,
here the bar means complex conjugate; and all inner products
involved in matrices are Hilbert-Schmidt one, i.e.,
$\Inner{\bsX}{\bsY} := \Tr{\bsX^\dagger\bsY}$, where $\dagger$ means
the complex conjugate and transpose.

In \cite{Byrd2003,Kimura2003}, the authors develop the Bloch vectors
formalism for an arbitrary finite-dimensional quantum system. In
this formalism, they use the following convention for the generators
$\bsG_k(k=1,\ldots,n^2-1)$, where $n\geqslant2$, of the Lie algebra
$\su(n)$ of $\SU(n)$:
\begin{eqnarray}
\Inner{\bsG_i}{\bsG_j}=2\delta_{ij}.
\end{eqnarray}
The commutation and anticommutation\footnote{Here
$[\bsX,\bsY]:=\bsX\bsY-\bsY\bsX$ stands for the commutator and
$\Set{\bsX,\bsY}:=\bsX\bsY+\bsY\bsX$ the anticommutator} for these
generators can also be specified as:
\begin{eqnarray}
[\bsG_i,\bsG_j] &=& 2\sqrt{-1}f_{ijk}\bsG_k = 2\sqrt{-1}f^{(k)}_{ij}\bsG_k\\
\Set{\bsG_i,\bsG_j} &=&
\frac4n\delta_{ij}\I+2d_{ijk}\bsG_k=\frac4n\delta_{ij}\I+2d^{(k)}_{ij}\bsG_k,
\end{eqnarray}
where the expansion coefficients
$d_{ijk}=\frac14\Tr{\set{\bsG_i,\bsG_j}\bsG_k}=:d^{(k)}_{ij}$ are
\emph{totally symmetric} and the expansion coefficients
$f_{ijk}=-\frac{\sqrt{-1}}4\Tr{[\bsG_i,\bsG_j]\bsG_k}=:f^{(k)}_{ij}$,
the structure constants of $\su(n)$, are \emph{totally
antisymmetric} in their indices; and by convention the summation is
performed for repeated indices. Both two identities can be
summarized as
\begin{eqnarray}\label{eq:GiGj}
\bsG_i\bsG_j =
\frac2n\delta_{ij}\I+\Pa{d^{(k)}_{ij}+\sqrt{-1}f^{(k)}_{ij}}\bsG_k,
\end{eqnarray}
which implies $\Tr{\bsG_i\bsG_j\bsG_k}=2\Pa{d_{ijk}+\sqrt{-1}f_{ijk}}=2\Pa{d^{(k)}_{ij}+\sqrt{-1}f^{(k)}_{ij}}$ or
\begin{eqnarray}
d_{ijk}=\frac12\re\Tr{\bsG_i\bsG_j\bsG_k}\quad\text{and}\quad f_{ijk}=\frac12\im\Tr{\bsG_i\bsG_j\bsG_k}.
\end{eqnarray}
These generators of the Lie algebra $\su(n)$ can be used to describe
any $n\times n$ qudit state (or density matrix) $\rho$ in terms of a
corresponding $(n^2-1)$-dimensional \emph{real} vector
$\bsr\in\bbR^{(n^2-1)}$:
\begin{eqnarray}
\rho(\bsr) = \frac1n\Pa{\I+\sqrt{\frac{n(n-1)}2}\bsr\cdot\bG},
\end{eqnarray}
where $\bsr=(r_1,\ldots,r_{n^2-1})\in\bbR^{(n^2-1)}$ and
$\bsr\cdot\bG=\sum^{n^2-1}_{k=1}r_k\bsG_k$ for
$\bG=(\bsG_1,\ldots,\bsG_{n^2-1})$. This representation is called
the \emph{generalized Gell-Mann representation} (or \emph{coherence
vector representation} with $\bsr$ the coherence vector
\cite{Byrd2003}), which are the higher-dimensional extensions of the
Pauli matrices $\boldsymbol{\sigma}:=(\sigma_1,\sigma_2,\sigma_3)$
on $\bbC^2$ and the Gell-Mann matrices on $\bbC^3$
\cite{Loubenets2021}. We use the notation
$\norm{\bsr}:=\sqrt{\sum^{n^2-1}_{k=1}r^2_k}$ stands for the usual
Euclidean length of $\bsr\in\bbR^{(n^2-1)}$. Denote by
\begin{eqnarray}
\Omega_n:=\Set{\bsr\in\bbR^{(n^2-1)}:\rho(\bsr)\geqslant\zero}.
\end{eqnarray}
We also denote by
\begin{eqnarray}
\Omega^{\text{ext}}_n:=\set{\bsr\in\Omega_n:\rho(\bsr)^2=\rho(\bsr)},
\end{eqnarray}
which corresponds to the set of all pure qudit states. As a submanifold, $\dim(\Omega^{\text{ext}}_n)=2(n-1)$. In order to
characterize analytically the structures of $\Omega_n$ and
$\Omega^{\text{ext}}_n$, respectively, we need the following
notion---symmetric star-product.
\begin{definition}[Symmetric star-product in $\bbR^{(n^2-1)}$, \cite{Byrd2003}]\label{def:starproduct}
The so-called \emph{symmetric star-product} $\star$ in
$\bbR^{(n^2-1)}$ on vectors $\bsx,\bsy\in\bbR^{(n^2-1)}$ through
\begin{eqnarray}
(\bsx\star\bsy)_k \defeq\frac1{n-2}
\sqrt{\frac{n(n-1)}2}d^{(k)}_{ij}x_iy_j,
\end{eqnarray}
where $k=1,\ldots,n^2-1$.
\end{definition}
Essentially, the symmetric star-product in $\bbR^{(n^2-1)}$ is a
diagonal mapping on $\bbR^{(n^2-1)}\times\bbR^{(n^2-1)}$. For a
convenient usage, we rewrite star-product as following form:
\begin{prop}\label{prop:starprod}
The star-product in $\bbR^{(n^2-1)}$ can be represented by
\begin{eqnarray}
\bsx\star\bsy \defeq
(\Innerm{\bsx}{\bsD_1}{\bsy},\ldots,\Innerm{\bsx}{\bsD_{n^2-1}}{\bsy})\quad(\forall\bsx,\bsy\in\bbR^{(n^2-1)}),
\end{eqnarray}
where $\bsD_k$ is an $(n^2-1)\times (n^2-1)$ real symmetric matrix
whose $(i,j)$-entries being identified with
\begin{eqnarray}\label{eq:Dk}
\Innerm{i}{\bsD_k}{j}:=\frac1{n-2}
\sqrt{\frac{n(n-1)}2}d^{(k)}_{ij},
\end{eqnarray}
where $k=1,\ldots,n^2-1$; and $i,j=1,\ldots,n^2-1$. Moreover we have
the following statements:
\begin{enumerate}[(i)]
\item $\Tr{\bsD_k}=0$ for all $k=1,\ldots,n^2-1$.
\item The set of all $\bsD_k$'s, $\Set{\bsD_k: k=1,\ldots,n^2-1}$, is linearly
independent for odd number $n$.
\item It holds that
\begin{eqnarray*}
\norm{\bsx\star\bsy}^2=\frac{n(n-1)}{8(n-2)^2}\Br{\Tr{(\bsx\cdot\bG\bsy\cdot\bG)^2)}+\frac4n\Inner{\bsx}{\bsx}\Inner{\bsy}{\bsy}+\frac{4(n-2)^2}{n(n-1)}\Inner{\bsx\star\bsx}{\bsy\star\bsy}-\frac8n\Inner{\bsx}{\bsy}^2}.
\end{eqnarray*}
In particular,
\begin{eqnarray*}
\norm{\bsx\star\bsx}^2=\frac{n(n-1)}{4(n-2)^2}\Br{\Tr{(\bsx\cdot\bG)^4}
- \frac4n\Inner{\bsx}{\bsx}^2}.
\end{eqnarray*}
\end{enumerate}
\end{prop}

\begin{proof}
The reformulation of star-product is trivially, and it is omitted
here. Next we show the 1st item. In fact, we have already known that
$\set{\I,\bsG_1,\ldots,\bsG_{n^2-1}}$ is orthogonal matrix basis for
the set of Hermitian matrices on $\bbC^n$. Then
$\set{\vec(\I),\vec(\bsG_1),\ldots,\vec(\bsG_{n^2-1})}$ is just
orthogonal basis for $\bbC^n\ot\bbC^n$. By normalizing them, we get
the orthonormal basis,
$$
\Set{\frac1{\sqrt{n}}\vec(\I),\frac1{\sqrt{2}}\vec(\bsG_1),\ldots,\frac1{\sqrt{2}}\vec(\bsG_{n^2-1})},
$$
leading to the following fact
\begin{eqnarray}\label{eq:generatorid}
\frac1n\vec(\I)\vec(\I)^\dagger+\frac12\sum^{n^2-1}_{k=1}\vec(\bsG_k)\vec(\bsG_k)^\dagger=\I\ot\I.
\end{eqnarray}
By partial-tracing the 2nd subsystem, we get that
$\frac1n\I+\frac12\sum^{n^2-1}_{k=1}\bsG^2_k=n\I$, that is,
\begin{eqnarray}
\sum^{n^2-1}_{k=1}\bsG^2_k=2\Pa{n-\frac1n}\I\propto\I.
\end{eqnarray}
This leads $\Tr{\bsD_k}=0$ for all $k=1,\ldots,n^2-1$. Indeed,
\begin{eqnarray*}
\Tr{\bsD_k}&=&\sum^{n^2-1}_{i=1}d^{(k)}_{ii} = \frac14\sum^{n^2-1}_{i=1}\Tr{\set{\bsG_i,\bsG_i}\bsG_k} = \frac12\Tr{\Pa{\sum^{n^2-1}_{i=1}\bsG^2_i}\bsG_k}\\
&=&\Pa{n-\frac1n}\Tr{\bsG_k}=0.
\end{eqnarray*}
For the 2nd item, in order to get the independence, we show that, if
there are $(n^2-1)$ real numbers $\alpha_k(k=1,\ldots,n^2-1)$ such
that $\sum^{n^2-1}_{k=1}\alpha_k\bsD_k=\zero$, then
$\alpha_k=0(k=1,\ldots,n^2-1)$. Indeed,
$\sum^{n^2-1}_{k=1}\alpha_k\bsD_k=\zero$ means that
\begin{eqnarray*}
0&=&\Innerm{i}{\sum^{n^2-1}_{k=1}\alpha_k\bsD_k}{j} =
\sum^{n^2-1}_{k=1}\alpha_kd^{(k)}_{ij}=\sum^{n^2-1}_{k=1}\alpha_kd^{(k)}_{ij}\\
&=&\frac12\sum^{n^2-1}_{k=1}\alpha_k\re\Tr{\bsG_i\bsG_j\bsG_k}=\frac12\re\Tr{\bsG_i\bsG_j\Pa{\sum^{n^2-1}_{k=1}\alpha_k\bsG_k}}.
\end{eqnarray*}
Denote
$\bsM:=\sum^{n^2-1}_{k=1}\alpha_k\bsG_k=\boldsymbol{\alpha}\cdot\bG$.
Apparently $\bsM$ is Hermitian matrix. From \eqref{eq:GiGj}, we see
that
\begin{eqnarray*}
\Tr{\set{\bsG_i,\bsG_j}\bsM}
=\Tr{\bsG_i\set{\bsG_j,\bsM}}=\Inner{\bsG_i}{\set{\bsG_j,\bsM}}=0\quad(\forall(i,j)).
\end{eqnarray*}
This means that $\set{\bsG_j,\bsM}\propto\I$ for all $j$. Without
loss of generality, let $\set{\bsG_j,\bsM}=c_j\I$ for some constants
$c_j\in\bbR$. By taking the trace on both sides, we get that
$c_j=\tfrac1n\Tr{\set{\bsG_j,\bsM}} = \tfrac4n \alpha_j$, leading to
\begin{eqnarray}
\set{\bsG_j,\bsM}=\frac4n \alpha_j\I,
\end{eqnarray}
which implies that
\begin{eqnarray*}
\frac4n \alpha_j\bsG_j &=&
\bsG_j\set{\bsG_j,\bsM}=\bsG^2_j\bsM+\bsG_j\bsM\bsG_j =
\bsG^2_j\bsM+\Pa{\frac4n \alpha_j\I-\bsM\bsG_j}\bsG_j\\
&=& \bsG^2_j\bsM+\frac4n \alpha_j\bsG_j - \bsM\bsG^2_j,
\end{eqnarray*}
i.e., $\bsG^2_j\bsM=\bsM\bsG^2_j$. Then $\bsM^2=
\frac2n\norm{\boldsymbol{\alpha}}^2\I$. Now from
\begin{eqnarray}\label{eq:generatorid2}
\bsM^2=(\boldsymbol{\alpha}\cdot\bG)^2 =
\frac2n\norm{\boldsymbol{\alpha}}^2\I+(n-2)\sqrt{\frac2{n(n-1)}}\boldsymbol{\alpha}\star\boldsymbol{\alpha}\cdot\bG,
\end{eqnarray}
we see that $\boldsymbol{\alpha}\star\boldsymbol{\alpha}=\zero$. For
a generic eigenvalue $\lambda$ of $\bsM$, it must be
$\lambda=\pm\sqrt{\tfrac2n}\norm{\boldsymbol{\alpha}}$. Let the
number of positive eigenvalues of $\bsM$ is $m$, then the number of
negative eigenvalue of $\bsM$ is $n-m$. By the fact that
$\Tr{\bsM}=0$, we see that
$$
m\sqrt{\tfrac2n}\norm{\boldsymbol{\alpha}}+(n-m)\Pa{-\sqrt{\tfrac2n}\norm{\boldsymbol{\alpha}}}=0,
$$
i.e., $(2m-n)\sqrt{\tfrac2n}\norm{\boldsymbol{\alpha}}=0$, which is
equivalent $2m-n=0$ or $\norm{\boldsymbol{\alpha}}=0$. If $n$ is
odd, we must have $\boldsymbol{\alpha}=\zero$ because $2m-n\neq0$.

Note that
\begin{eqnarray*}
\norm{\bsx\star\bsy}^2&=& \sum^{n^2-1}_{k=1}\Innerm{\bsx}{\bsD_k}{\bsy}^2 = \frac1{(n-2)^2}\frac{n(n-1)}2\sum^{n^2-1}_{k=1}\Br{\sum_{i,j}x_i d^{(k)}_{ij}y_j}^2\\
&=&\frac1{16}\frac1{(n-2)^2}\frac{n(n-1)}2\sum^{n^2-1}_{k=1}\Br{\Tr{\set{\bsx\cdot\bG,\bsy\cdot\bG}\bsG_k}}^2\\
&=&\frac{n(n-1)}{32(n-2)^2}\Innerm{\vec(\set{\bsx\cdot\bG,\bsy\cdot\bG})}{\sum^{n^2-1}_{k=1}\vec(\bsG_k)\vec(\bsG_k)^\dagger}{\vec(\set{\bsx\cdot\bG,\bsy\cdot\bG})}.
\end{eqnarray*}
Using Eq.~\eqref{eq:generatorid}, we get that
\begin{eqnarray*}
\norm{\bsx\star\bsy}^2&=& \frac{n(n-1)}{32(n-2)^2}\Innerm{\vec(\set{\bsx\cdot\bG,\bsy\cdot\bG})}{2\I\ot\I - \frac2n\vec(\I)\vec(\I)^\dagger}{\vec(\set{\bsx\cdot\bG,\bsy\cdot\bG})}\\
&=&\frac{n(n-1)}{16(n-2)^2}\Br{\Tr{\set{\bsx\cdot\bG,\bsy\cdot\bG}^2} -\frac1n \Pa{\Tr{\set{\bsx\cdot\bG,\bsy\cdot\bG}}}^2}\\
&=&\frac{n(n-1)}{8(n-2)^2}\Br{\Tr{(\bsx\cdot\bG\bsy\cdot\bG)^2}+\Tr{(\bsx\cdot\bG)^2(\bsy\cdot\bG)^2}
- \frac2n\Pa{\Tr{(\bsx\cdot\bG)(\bsy\cdot\bG)}}^2}.
\end{eqnarray*}
Now we see that
$\Tr{(\bsx\cdot\bG)^2(\bsy\cdot\bG)^2}=\frac4n\Inner{\bsx}{\bsx}\Inner{\bsy}{\bsy}+\frac{4(n-2)^2}{n(n-1)}\Inner{\bsx\star\bsx}{\bsy\star\bsy}$.
It follows from that
\begin{eqnarray*}
\norm{\bsx\star\bsy}^2=\frac{n(n-1)}{8(n-2)^2}\Br{\Tr{(\bsx\cdot\bG\bsy\cdot\bG)^2)}+\frac4n\Inner{\bsx}{\bsx}\Inner{\bsy}{\bsy}+\frac{4(n-2)^2}{n(n-1)}\Inner{\bsx\star\bsx}{\bsy\star\bsy}-\frac8n\Inner{\bsx}{\bsy}^2}.
\end{eqnarray*}
In particular, for $\bsx=\bsy$, we see that
\begin{eqnarray*}
\norm{\bsx\star\bsx}^2&=&
\frac{n(n-1)}{4(n-2)^2}\Br{\Tr{(\bsx\cdot\bG)^4} -
\frac1n\Pa{\Tr{(\bsx\cdot\bG)^2}}^2}\\
&=& \frac{n(n-1)}{4(n-2)^2}\Br{\Tr{(\bsx\cdot\bG)^4} -
\frac4n\Inner{\bsx}{\bsx}^2}.
\end{eqnarray*}
This completes the proof.
\end{proof}

\begin{prop}
For any positive integer $n>2$, it holds that
\begin{eqnarray}
\Omega^{\mathrm{ext}}_n = \Set{\bsr\in\bbR^{(n^2-1)}:
\norm{\bsr}=1,\bsr\star\bsr=\bsr}.
\end{eqnarray}
In particular, for $n=2$, $\Omega^{\mathrm{ext}}_2 = \Set{\bsr\in\bbR^3:
\norm{\bsr}=1}$.
\end{prop}

\begin{proof}
Recall that
\begin{eqnarray}
(\bsr\cdot\bG)^2 =
\frac2n\norm{\bsr}^2\I+(n-2)\sqrt{\frac2{n(n-1)}}\bsr\star\bsr\cdot\bG.
\end{eqnarray}
Note that
\begin{eqnarray*}
\rho(\bsr)^2 &=&
\frac1{n^2}\Pa{\I+2\sqrt{\frac{n(n-1)}2}\bsr\cdot\bG +
\frac{n(n-1)}2(\bsr\cdot\bG)^2}\\
&=&\frac1{n^2}\Br{\Pa{1+(n-1)\norm{\bsr}^2}\I+\sqrt{\frac{n(n-1)}2}\Pa{2\bsr+(n-2)\bsr\star\bsr}\cdot\bG},\\
\rho(\bsr)&=&\frac1{n^2}\Pa{n\I+n\sqrt{\frac{n(n-1)}2}\bsr\cdot\bG}.
\end{eqnarray*}
Then $\rho(\bsr)^2=\rho(\bsr)$ becomes as
\begin{eqnarray}
1+(n-1)\norm{\bsr}^2&=&n\\
\sqrt{\frac{n(n-1)}2}\Pa{2\bsr+(n-2)\bsr\star\bsr}&=&n\sqrt{\frac{n(n-1)}2}\bsr
\end{eqnarray}
That is, $\norm{\bsr}=1$ and $\bsr\star\bsr=\bsr$ for $n>2$; or
$\norm{\bsr}=1$ for $n=2$.
\end{proof}

%\begin{remark}
%From the above discussion, denote $p_k:=\Tr{\rho(\bsr)^k}$, we see
%that
%\begin{itemize}
%\item $p_2=\frac{1+(n-1)\norm{\bsr}^2}n$
%\item $p_3=\frac{1+3(n-1)\norm{\bsr}^2+(n-1)(n-2)\Inner{\bsr}{\bsr\star\bsr}}{n^2}$
%\item $p_4=\frac{1+6(n-1)\norm{\bsr}^2+4(n-1)(n-2)\Inner{\bsr}{\bsr\star\bsr}+(n-1)^2\norm{\bsr}^4+(n-1)(n-2)^2\Inner{\bsr\star\bsr}{\bsr\star\bsr}}{n^3}$
%\end{itemize}
%\end{remark}

%=================================================================%
\section{Main result: additive uncertainty relation}\label{sect:3}
%=================================================================%

Given two qudit observables $\bsA=a_0\I+\bsa\cdot\bG$ and
$\bsB=b_0\I+\bsb\cdot\bG$, acting on $\bbC^n$. For any qudit state
$\rho=\frac1n\Pa{\I+\sqrt{\frac{n(n-1)}2}\bsr\cdot\bG}$, we see that
\begin{eqnarray*}
\Inner{\bsA}{\rho} = a_0 +
\sqrt{\tfrac{2(n-1)}n}\Inner{\bsa}{\bsr},\quad \Inner{\bsB}{\rho} =
b_0 + \sqrt{\tfrac{2(n-1)}n}\Inner{\bsb}{\bsr}.
\end{eqnarray*}
The so-called additive uncertainty relation is the following matrix
optimization problem:
\begin{eqnarray}
\var_\rho(\bsA)+\var_\rho(\bsB)\geqslant
\min_{\rho\in\density{\bbC^n}}\Br{\var_\rho(\bsA)+\var_\rho(\bsB)}=:m_{\bsA,\bsB}.
\end{eqnarray}
Here $\var_\rho(\bsX)$ is the matrix variance of observable $\bsX$,
defined by $\var_\rho(\bsX):=\Inner{\bsX^2}{\rho} -
\Inner{\bsX}{\rho}^2$, where $\bsX=\bsA,\bsB$ . Because $\rho\mapsto
\var_\rho(\bsA)+\var_\rho(\bsB)$ is \emph{concave} in the argument
$\rho$, we see that
\begin{eqnarray}
m_{\bsA,\bsB}=\min_{\ket{\psi}\in\bbC^n:\norm{\psi}=1}\Pa{\var_\psi(\bsA)+\var_\psi(\bsB)}.
\end{eqnarray}
Note that
\begin{eqnarray*}
\bsA^2 =
\Pa{a^2_0+\frac2n\norm{\bsa}^2}\I+2a_0\bsa\cdot\bG+(n-2)\sqrt{\tfrac2{n(n-2)}}\bsa\star\bsa\cdot\bG.
\end{eqnarray*}
We get that
\begin{eqnarray*}
\Inner{\bsA^2}{\rho}
=a^2_0+\frac2n\norm{\bsa}^2+\Inner{2a_0\sqrt{\tfrac{2(n-1)}{n}}\bsa+\tfrac{2(n-2)}n\bsa\star\bsa}{\bsr}.
\end{eqnarray*}
From the above observation, we see that
\begin{eqnarray}
&&\var_\rho(\bsA)+\var_\rho(\bsB) = \Inner{\bsA^2+\bsB^2}{\rho}
-\Inner{\bsA}{\rho}^2-\Inner{\bsB}{\rho}^2\notag\\
&&=\frac2n\Pa{\norm{\bsa}^2+\norm{\bsb}^2}+\frac{2(n-2)}{n}\Inner{\bsa\star\bsa+\bsb\star\bsb}{\bsr}-\frac{2(n-1)}{n}\Pa{\Inner{\bsa}{\bsr}^2+\Inner{\bsb}{\bsr}^2}.
\end{eqnarray}
Let
\begin{eqnarray}\label{eq:lab}
\ell_{\bsa,\bsb}:=\min_{\bsr\in\Omega^{\text{ext}}_n}\Br{(n-2)\Inner{\bsa\star\bsa+\bsb\star\bsb}{\bsr}-(n-1)\Pa{\Inner{\bsa}{\bsr}^2+\Inner{\bsb}{\bsr}^2}}.
\end{eqnarray}
Thus we get that
\begin{eqnarray}
m_{\bsA,\bsB}=\frac2n(\norm{\bsa}^2+\norm{\bsb}^2+\ell_{\bsa,\bsb}).
\end{eqnarray}
In the following, we consider to calculate $\ell_{\bsa,\bsb}$. In
fact,

\begin{thrm}\label{th:main}
For given two qudit observables $\bsA=a_0\I+\bsa\cdot\bG$ and
$\bsB=b_0\I+\bsb\cdot\bG$ on $\bbC^n$, let
\begin{eqnarray}
\bsT_{\bsa,\bsb}:=(n-2)\sum^{n^2-1}_{k=1}\Tr{\bsO_{\bsa,\bsb}\bsD_k}\bsD_k-(n-1)\bsO_{\bsa,\bsb},
\end{eqnarray}
where $\bsO_{\bsa,\bsb}:=\proj{\bsa}+\proj{\bsb}$ and $\bsD_k$'s are identified
from \eqref{eq:Dk}. It holds that
\begin{eqnarray}\label{eq:lab}
\ell_{\bsa,\bsb} &=&
\min_{\bsr\in\Omega^{\mathrm{ext}}_n}\Innerm{\bsr}{\bsT_{\bsa,\bsb}}{\bsr}\\
m_{\bsA,\bsB} &=& \frac2n\Br{-\frac1{n-1}\Tr{\bsT_{\bsa,\bsb}}+\min_{\bsr\in\Omega^{\mathrm{ext}}_n}\Innerm{\bsr}{\bsT_{\bsa,\bsb}}{\bsr}}.
\end{eqnarray}
\end{thrm}

\begin{proof}
Because $\bsr\in\Omega^{\text{ext}}_n$, we see $\bsr=\bsr\star\bsr$.
Then
\begin{eqnarray*}
\Inner{\bsa\star\bsa+\bsb\star\bsb}{\bsr} &=&
\Inner{\bsa\star\bsa+\bsb\star\bsb}{\bsr\star\bsr}\\
&=&\sum^{n^2-1}_{k=1}(\Innerm{\bsa}{\bsD_k}{\bsa}+\Innerm{\bsb}{\bsD_k}{\bsb})\Innerm{\bsr}{\bsD_k}{\bsr}\\
&=&\sum^{n^2-1}_{k=1}\Tr{\bsO_{\bsa,\bsb}\bsD_k}\Innerm{\bsr}{\bsD_k}{\bsr}\\
&=&
\Innerm{\bsr}{\sum^{n^2-1}_{k=1}\Tr{\bsO_{\bsa,\bsb}\bsD_k}\bsD_k}{\bsr}
\end{eqnarray*}
and $\Inner{\bsa}{\bsr}^2+\Inner{\bsb}{\bsr}^2 =
\Innerm{\bsr}{\bsO_{\bsa,\bsb}}{\bsr}$. Therefore
\begin{eqnarray*}
\ell_{\bsa,\bsb} &=&
\min_{\bsr\in\Omega^{\text{ext}}_n}\Br{(n-2)\Inner{\bsa\star\bsa+\bsb\star\bsb}{\bsr}-(n-1)\Pa{\Inner{\bsa}{\bsr}^2+\Inner{\bsb}{\bsr}^2}}\\
&=&
\min_{\bsr\in\Omega^{\text{ext}}_n}\Br{(n-2)\Inner{\bsa\star\bsa+\bsb\star\bsb}{\bsr\star\bsr}-(n-1)\Pa{\Inner{\bsa}{\bsr}^2+\Inner{\bsb}{\bsr}^2}}\\
\\
&=&\min_{\bsr\in\Omega^{\mathrm{ext}}_n}\Br{\Innerm{\bsr}{(n-2)\sum^{n^2-1}_{k=1}\Tr{\bsO_{\bsa,\bsb}\bsD_k}\bsD_k}{\bsr}-(n-1)\Innerm{\bsr}{\bsO_{\bsa,\bsb}}{\bsr}}\\
&=&\min_{\bsr\in\Omega^{\mathrm{ext}}_n}\Innerm{\bsr}{\bsT_{\bsa,\bsb}}{\bsr}.
\end{eqnarray*}
Note that $\Tr{\bsD_k}=0$ for all $k=1,\ldots,n^2-1$. This implies
that
$\Tr{\bsT_{\bsa,\bsb}}=-(n-1)\Tr{\bsO_{\bsa,\bsb}}=-(n-1)(\norm{\bsa}^2+\norm{\bsb}^2)$,
i.e.,
$\norm{\bsa}^2+\norm{\bsb}^2=-\frac1{n-1}\Tr{\bsT_{\bsa,\bsb}}$.
Substituting this expression into the defining equation of
$m_{\bsA,\bsB}$, we get the desired result. This completes the
proof.
\end{proof}

In fact, our method here can also be used to study the additive
uncertainty relation for multiple qudit observables. According to
the above reasoning, we get the following result:
\begin{cor}
For any given $K$-tuple of qudit observables
$(\bsA_1,\ldots,\bsA_K)$ on $\bbC^n$, where $\bsA_\mu=a^{(\mu)}_0\I+\bsa_\mu\cdot\bG$
for $\mu=1,\ldots,K$. Let
\begin{eqnarray}
m_{\bsA_1,\ldots,\bsA_K}&:=&\min_{\rho\in\density{\bbC^n}}\sum^K_{\mu=1}\var_\rho(\bsA_\mu),\\
\bsO_{\bsa_1,\ldots,\bsa_K}&:=&\sum^K_{\mu=1}\proj{\bsa_\mu},\\
\bsT_{\bsa_1,\ldots,\bsa_K}&:=&(n-2)\sum^{n^2-1}_{k=1}\Tr{\bsO_{\bsa_1,\ldots,\bsa_K}\bsD_k}\bsD_k-(n-1)\bsO_{\bsa_1,\ldots,\bsa_K},\\
\ell_{\bsa_1,\ldots,\bsa_K}&:=&
\min_{\bsr\in\Omega^{\mathrm{ext}}_n}\Innerm{\bsr}{\bsT_{\bsa_1,\ldots,\bsa_K}}{\bsr},
\end{eqnarray}
where $\bsD_j$'s are identified from \eqref{eq:Dk}. It holds that
\begin{eqnarray}
m_{\bsA_1,\ldots,\bsA_K} &=&
\frac2n\Pa{\sum^K_{\mu=1}\norm{\bsa_\mu}^2+\ell_{\bsa_1,\ldots,\bsa_K}}\\
&=& \frac2n\Br{-\frac1{n-1}\Tr{\bsT_{\bsa_1,\ldots,\bsa_K}}+\min_{\bsr\in\Omega^{\mathrm{ext}}_n}\Innerm{\bsr}{\bsT_{\bsa_1,\ldots,\bsa_K}}{\bsr}}.
\end{eqnarray}
\end{cor}

\begin{proof}
The proof can be easily obtained by repeating $K$ times the proof of
Theorem~\ref{th:main}. It is omitted here.
\end{proof}

\begin{remark}\label{rem:3.3}
From the above theorem, we can see that the optimization problem
\eqref{eq:lab} is the famous quadratic programming
\cite{Nocedal2006,Best2017,Dostal2009} with the following prototype:
$\min_{\bsr\in\Omega^{\text{ext}}_n}\Innerm{\bsr}{\bsT}{\bsr}$,
where $\bsT$ is an $(n^2-1)\times (n^2-1)$ real symmetric matrix.
Namely,
\begin{center}
\fcolorbox{purple}{lightgray}{
\parbox{7cm}{
\begin{eqnarray*}
\min \Innerm{\bsr}{\bsT}{\bsr}\\
\text{subject to: } \norm{\bsr}=1,\\
\bsr\star\bsr=\bsr.
\end{eqnarray*}}
}
\end{center}
In Appendix~\ref{app:algorithm}, we present a universal algorithm to
find the minimum of $\Innerm{\bsr}{\bsT}{\bsr}$, where
$\bsr\in\Omega^{\mathrm{ext}}_n$ for $n=3$.

In addition, such problem is also related to the numerical range of
$\bsT$, defined by $\set{\Innerm{\bsr}{\bsT}{\bsr}:\norm{\bsr}=1}$.
We have already known that the numerical range of $\bsT$ is just the
closed interval $[\lambda_{\min}(\bsT),\lambda_{\max}(\bsT)]$, where
$\lambda_{\max/\min}(\bsT)$ is the maximal/minimal eigenvalue of
$\bsT$. Clearly the set
$\set{\Innerm{\bsr}{\bsT}{\bsr}:\bsr\in\Omega^{\text{ext}}_n}$ is
also a closed subinterval of the numerical range of $\bsT$, called
the \emph{constrained numerical range}. The minimal boundary point
of $\set{\Innerm{\bsr}{\bsT}{\bsr}:\bsr\in\Omega^{\text{ext}}_n}$ is
the desired one.
\end{remark}

\begin{remark}
Let
\begin{eqnarray}
\cU_{\text{m}}(n,K):=\Set{\bsr\in\Omega^{\text{ext}}_n:\Innerm{\bsr}{\bsT}{\bsr} = \ell_{\bsa_1,\ldots,\bsa_K}}.
\end{eqnarray}
Using such set, we can construct the set of all qudit states of
\emph{minimal uncertainty} in the sense of
$\min_{\rho\in\density{\complex^n}}[\var_\rho(\bsA)+\var_\rho(\bsB)]=m_{\bsA,\bsB}$:
\begin{eqnarray}
\cD_{\text{m}}(n,K):=\Set{\rho(\bsr):\bsr\in \cU_{\text{m}}(n,K)}.
\end{eqnarray}
\end{remark}

%======================================================================%
\section{Typical example I: the qubit state case}\label{sect:4}
%======================================================================%

The Pauli matrices are the generator of $\su(2)$, which are given by the following:
\begin{eqnarray}
\bsG_1=\Pa{\begin{array}{cc}
0&1\\
1&0\\
\end{array}},\quad \bsG_2=\Pa{\begin{array}{cc}
0&-\mathrm{i}\\
\mathrm{i}&0\\
\end{array}},\quad \bsG_3=\Pa{\begin{array}{cc}
1&0\\
0&-1\\
\end{array}}.
\end{eqnarray}
Instead of using $\bsG_k$'s, we use the notation $\sigma_k$ to replace $\bsG_k$ in the qubit system. The Pauli matrices satisfy the identities:
\begin{eqnarray}
\sigma^2_k&=&\I,\\
\set{\sigma_i,\sigma_j}&=&2\delta_{ij}\I,\\
\sigma_i\sigma_j&=&\delta_{ij}\I+\sqrt{-1}\epsilon_{ijk}\sigma_k,
\end{eqnarray}
where $\epsilon$ is the permutation symbol. These facts leads to the fact that $\bsD_1=\bsD_2=\bsD_3=\zero$, i.e., $\bsx\star\bsy=\zero$ for $\bsx,\bsy\in\bbR^3$.

For any qubit observables $\bsA=a_0\I+\bsa\cdot\boldsymbol{\sigma}$ and $\bsB=b_0\I+\bsb\cdot\boldsymbol{\sigma}$, we see from Theorem~\ref{th:main} that $\bsT_{\bsa,\bsb}=-\bsO_{\bsa,\bsb}$ and
\begin{eqnarray*}
m_{\bsA,\bsB} &=& \norm{\bsa}^2+\norm{\bsb}^2+\min_{\bsr\in\Omega^{\text{ext}}_2}\Innerm{\bsr}{\bsT_{\bsa,\bsb}}{\bsr}\\
&=& \Tr{\bsO_{\bsa,\bsb}}-\max\set{\Innerm{\bsr}{\bsO_{\bsa,\bsb}}{\bsr}:\norm{\bsr}=1}\\
&=& \Tr{\bsO_{\bsa,\bsb}}-\lambda_{\max}(\bsO_{\bsa,\bsb}),
\end{eqnarray*}
where $\lambda_{\max}(\bsO_{\bsa,\bsb})$ is attained at any one element in the following set
\begin{eqnarray}
\cU_{\text{m}}(2,2):=\Set{\bsr\in\Omega^{\text{ext}}_2:\Innerm{\bsr}{\bsO_{\bsa,\bsb}}{\bsr}=\lambda_{\max}(\bsO_{\bsa,\bsb})}.
\end{eqnarray}
In fact, in the above set, we can choose the eigenvector $\tilde \bsr=(r_1,r_2,r_3)$ of $\bsO_{\bsa,\bsb}$ corresponding to the eigenvalue $\lambda_{\max}(\bsO_{\bsa,\bsb})$, i.e., $\bsO_{\bsa,\bsb}\ket{\tilde\bsr}=\lambda_{\max}(\bsO_{\bsa,\bsb})\ket{\tilde\bsr}$. This means that
\begin{eqnarray}
\min_{\psi\in\bbC^2:\norm{\psi}=1}\Br{\var_\psi(\bsA)+\var_\psi(\bsB)} = \var_{\tilde\psi}(\bsA)+\var_{\tilde\psi}(\bsB),
\end{eqnarray}
where $\tilde\psi=\frac12(\I+\tilde\bsr\cdot\boldsymbol{\sigma})$, which is the state of the so-called \emph{minimal uncertainty} in the sense above. All states of minimal uncertainty is just the following set
\begin{eqnarray}
\cD_{\text{m}}(2,2)=\Set{\psi(\bsr):\bsr\in\cU_{\text{m}}(2,2)}.
\end{eqnarray}
Besides, $\bsO_{\bsa,\bsb}$ is unitarily equivalent to the $2\times 2$ matrix $\Pa{\begin{array}{cc}
\Inner{\bsa}{\bsa}&\Inner{\bsa}{\bsb}\\
\Inner{\bsa}{\bsb}&\Inner{\bsb}{\bsb}\\
\end{array}}$ whose minimal eigenvalue is just $\Tr{\bsO_{\bsa,\bsb}}-\lambda_{\max}(\bsO_{\bsa,\bsb})$. This recovers the result in \cite[Theorem~2.2]{Zhang2021}.

For any finite number of qubit observables $\bsA_k=a^{(k)}_0\I+\bsa_k\cdot\boldsymbol{\sigma}$, where $k=1,\ldots,K$, we have
\begin{eqnarray*}
m_{\bsA_1,\ldots,\bsA_K} = \Tr{\bsO_{\bsa_1,\ldots,\bsa_K}}-\lambda_{\max}(\bsO_{\bsa_1,\ldots,\bsa_K}).
\end{eqnarray*}
where $\lambda_{\max}(\bsO_{\bsa_1,\ldots,\bsa_K})$ is attained at
any one element in the following set
\begin{eqnarray}
\cU_{\text{m}}(2,K):=\Set{\bsr\in\Omega^{\text{ext}}_2:\Innerm{\bsr}{\bsO_{\bsa_1,\ldots,\bsa_K}}{\bsr}=\lambda_{\max}(\bsO_{\bsa_1,\ldots,\bsa_K})}.
\end{eqnarray}
All states of minimal uncertainty is just the following set
\begin{eqnarray}
\cD_{\text{m}}(2,K)=\Set{\psi(\bsr):\bsr\in\cU_{\text{m}}(2,K)}.
\end{eqnarray}

%====================================================================%
\section{Typical example II: the qutrit state case}\label{sect:5}
%====================================================================%

The so-called \emph{Gell-Mann matrices} are the following $3\times3$
Hermitian matrices, given by
\begin{eqnarray*}
&&\bsG_1 =\Pa{\begin{array}{ccc}
0&1&0\\
1&0&0\\
0&0&0\\
\end{array}}, \bsG_2=\Pa{\begin{array}{ccc}
0&-\mathrm{i}&0\\
\mathrm{i}&0&0\\
0&0&0\\
\end{array}},\bsG_3 = \Pa{\begin{array}{ccc}
1&0&0\\
0&-1&0\\
0&0&0\\
\end{array}},\bsG_4=\Pa{\begin{array}{ccc}
0&0&1\\
0&0&0\\
1&0&0\\
\end{array}},\\&&
\bsG_5=\Pa{\begin{array}{ccc}
0&0&-\mathrm{i}\\
0&0&0\\
\mathrm{i}&0&0\\
\end{array}},\bsG_6=\Pa{\begin{array}{ccc}
0&0&0\\
0&0&1\\
0&1&0\\
\end{array}},\bsG_7 = \Pa{\begin{array}{ccc}
0&0&0\\
0&0&-\mathrm{i}\\
0&\mathrm{i}&0\\
\end{array}},\bsG_8=\Pa{
\begin{array}{ccc}
\tfrac1{\sqrt{3}}&0&0\\
0&\tfrac1{\sqrt{3}}&0\\
0&0&-\tfrac2{\sqrt{3}}\\
\end{array}}.
\end{eqnarray*}

Some useful properties of Gell-Mann matrices are listed below.
\begin{itemize}
\item $\Tr{\bsG_k}=0$, where $k=1,\ldots,8$;
\item $\Inner{\bsG_i}{\bsG_j}=2\delta_{ij}$, where $1\leqslant i,j\leqslant8$;
\item $\Set{\I,\bsG_1,\ldots,\bsG_8}$ is an orthogonal matrix
basis such that each $3\times 3$ Hermitian matrix $\bsA$ can be
represented by
\begin{eqnarray*}
\bsA = a_0\I+\bsa\cdot\bG,\quad (a_0,\bsa)\in\bbR^9,
\end{eqnarray*}
where $\bsa\cdot\bG:=\sum^8_{k=1}a_k\bsG_k$ for
$\bG:=(\bsG_1,\ldots,\bsG_8)$.
\item $\Inner{\bsa\cdot\bG}{\bsb\cdot\bG}=2\Inner{\bsa}{\bsb}$, where $\bsa,\bsb\in\bbR^8$.
\item If $3\times 3$ Hermitian matrix $\bsA=a_0\I+\bsa\cdot\bG$, then we
have
\begin{eqnarray*}
a_0=\frac13\Tr{\bsA},\quad
a_k=\frac12\Inner{\bsA}{\bsG_k}\quad(k=1,\ldots,8).
\end{eqnarray*}
\item The algebraic structure of these matrices is determined by the
product property:
\begin{eqnarray*}
\bsG_i\bsG_j=\tfrac23\delta_{ij}\I+(d^{(k)}_{ij}+\mathrm{i}f^{(k)}_{ij})\bsG_k,
\end{eqnarray*}
where the expansion coefficients $d_{ijk}:=d^{(k)}_{ij}$ are
\emph{totally symmetric}. The numerical values of all the
independent non-vanishing components of $d_{ijk}$ are given by:
\begin{eqnarray*}
d_{118}=d_{228}=d_{338}=-d_{888}=\frac1{\sqrt{3}},\\
d_{146}=d_{157}=-d_{247}=d_{256}=\frac12,\\
d_{344}=d_{355}=-d_{366}=-d_{377}=\frac12,\\
d_{448}=d_{558}=d_{668}=d_{778}=-\frac1{2\sqrt{3}},
\end{eqnarray*}
and the expansion coefficients $f_{ijk}:=f^{(k)}_{ij}$, the
structure constants of the Lie algebra of $\SU(3)$, are
\emph{totally antisymmetric} in their indices. The numerical values
of all the independent non-vanishing components of $f_{ijk}$ are
given by:
\begin{eqnarray*}
f_{123}=1,\quad f_{458}=f_{678}=\frac{\sqrt{3}}2,\\
f_{147}=f_{246}=f_{257}=f_{345}=f_{516}=f_{637}=\frac12.
\end{eqnarray*}
\end{itemize}

For a convenient usage, we rewrite
star-product as following form:
\begin{prop}\label{prop:starprod8}
For $n=3$, the star-product in $\bbR^8$ can be represented by
\begin{eqnarray*}
\bsx\star\bsy =
(\Innerm{\bsx}{\bsD_1}{\bsy},\ldots,\Innerm{\bsx}{\bsD_8}{\bsy})\quad(\forall\bsx,\bsy\in\bbR^8),
\end{eqnarray*}
where $\bsD_k(k=1,\ldots,8)$ are displayed in
Appendix~\ref{app:3Dk}. Moreover, $\bsx\star\bsy=\bsy\star\bsx$ and
$$
\Inner{\bsx\star\bsx}{\bsx\star\bsx}=\Inner{\bsx}{\bsx}^2\Longleftrightarrow\norm{\bsx\star\bsx}=\norm{\bsx}^2.
$$
\end{prop}
For a completeness, here we present a proof again albeit they
appeared in other literatures.
\begin{proof}
In fact, $\bsx\star\bsy =
(\Innerm{\bsx}{\bsD_1}{\bsy},\ldots,\Innerm{\bsx}{\bsD_8}{\bsy})$
can be derived from the Definition~\ref{def:starproduct} of the
symmetric star-product. Due to the symmetry of all matrices
$\bsD_k$'s,
$\Innerm{\bsx}{\bsD_k}{\bsy}=\Innerm{\bsy}{\bsD_k}{\bsx}$ for all
$k=1,\ldots,8$. This indicates that $\bsx\star\bsy=\bsy\star\bsx$.
Finally, we see from Proposition~\ref{prop:starprod} (iii) that
\begin{eqnarray*}
\norm{\bsx\star\bsx}^2&=&\frac32\Tr{(\bsx\cdot\bG)^4} -
2\Inner{\bsx}{\bsx}^2\\
&=&
3\Inner{\bsx}{\bsx}^2-2\Inner{\bsx}{\bsx}^2=\Inner{\bsx}{\bsx}^2,
\end{eqnarray*}
where we use the Gell-Mann matrices to derive the fact
$\Tr{(\bsx\cdot\bG)^4}=2\Inner{\bsx}{\bsx}^2$ in the second
equality, that is, $\norm{\bsx\star\bsx}=\norm{\bsx}^2$. This
completes the proof.
\end{proof}
Some properties of star product can also be found in
\cite{Arvind1997,Ercolessi2001,Goyal2016}.

\begin{lem}\label{lem:pos}
For a generic $3\times 3$ Hermitian matrix $\rho$ of fixed-trace
one, it is positive semi-definite if and only if
$\Tr{\rho^2}\leqslant1$ and $\det(\rho)\geqslant0$.
\end{lem}

\begin{proof}
Let $\lambda_k$ be all eigenvalues of $\rho$ with
$\lambda_1\geqslant\lambda_2\geqslant\lambda_3$. Then we have that
$\lambda_1+\lambda_2+\lambda_3=1$, which leads to the result
$\lambda_1\geqslant\frac13\geqslant\lambda_3$. Now
$\Tr{\rho^2}\leqslant1$ means that
$\lambda^2_1+\lambda^2_2+\lambda^2_3\leqslant1$, which implies that
$\lambda^2_k\leqslant1$, where $k=1,2,3$. Thus
$1\geqslant\lambda_1\geqslant\frac13$, thus
$\lambda_2+\lambda_3=1-\lambda_1\in[0,\frac23]$. If $\lambda_1=1$,
then $\lambda_2=\lambda_3=0$ and thus $\rho\geqslant\zero$; if
$\lambda_1\in[\frac13,1)$, then
$\lambda_2+\lambda_3=1-\lambda_1\in(0,\frac23]$, together with
$\lambda_2\geqslant\lambda_3$, implying that $\lambda_2>0$. Now
$\det(\rho)=\lambda_1\lambda_2\lambda_3\geqslant0$ amounts to say
that $\lambda_3\geqslant0$. Therefore $\lambda_k\geqslant0$ for all
$k=1,2,3$. That is, $\rho\geqslant\zero$.

Based on the above observation, we get that $\rho\geqslant\zero$
whenever $\Tr{\rho^2}\leqslant1$ and $\det(\rho)\geqslant0$ for the
$3\times 3$ Hermitian matrix $\rho$ of fixed-trace one.
\end{proof}

With these preparations, we can now derive very quickly the
characterizations of $\Omega_3$ and $\Omega^{\text{ext}}_3$.

\begin{prop}
It holds that
\begin{eqnarray}
\Omega_3&=&\Set{\bsr\in\bbR^8: \norm{\bsr}\leqslant1,
1+2\Inner{\bsr}{\bsr\star\bsr}\geqslant3\Inner{\bsr}{\bsr}},\\
\Omega^{\mathrm{ext}}_3 &=&\Set{\bsr\in\bbR^8: \norm{\bsr}=1,
\bsr\star\bsr=\bsr}.
\end{eqnarray}
\end{prop}

\begin{proof}
In fact, for the parametrization,
$\rho(\bsr)=\frac13(\I+\sqrt{3}\bsr\cdot\bG)$, we see that
\begin{eqnarray}
\Tr{\rho(\bsr)^2}&=& \frac{1+2\norm{\bsr}^2}3,\\
\det(\rho(\bsr)) &=&
\frac{1+2\Inner{\bsr}{\bsr\star\bsr}-3\Inner{\bsr}{\bsr}}{27}.
\end{eqnarray}
Thus by Lemma~\ref{lem:pos}, $\Tr{\rho(\bsr)^2}\leqslant1$ is
equivalent to $\norm{\bsr}\leqslant1$, and
$\det(\rho(\bsr))\geqslant0$ is equivalent to
$1+2\Inner{\bsr}{\bsr\star\bsr}\geqslant3\Inner{\bsr}{\bsr}$. Put
together, we get that
$$
\Omega_3=\Set{\bsr\in\bbR^8: \norm{\bsr}\leqslant1,
1+2\Inner{\bsr}{\bsr\star\bsr}\geqslant3\Inner{\bsr}{\bsr}}.
$$
Now for $\bsr\in\Omega^{\mathrm{ext}}_3$, we get that
$\bsr\in\Omega_3$ and $\Tr{\rho(\bsr)^2}=1$. That is,
$\bsr\in\Omega_3$ and $\norm{\bsr}=1$. Substituting $\norm{\bsr}=1$
into $1+2\Inner{\bsr}{\bsr\star\bsr}\geqslant3\Inner{\bsr}{\bsr}=3$,
reduced to $\Inner{\bsr}{\bsr\star\bsr}\geqslant1$. The using
Cauchy-Schwarz inequality, we see that
\begin{eqnarray}
1\leqslant \Inner{\bsr}{\bsr\star\bsr}\leqslant
\norm{\bsr}\norm{\bsr\star\bsr}=\norm{\bsr}^3=1.
\end{eqnarray}
This inequality is saturated iff $\bsr\propto\bsr\star\bsr$. The
following facts $\norm{\bsr}=1$,
$\norm{\bsr\star\bsr}=\norm{\bsr}^2=1$, and
$\Inner{\bsr}{\bsr\star\bsr}=1$ indicate that $\bsr=\bsr\star\bsr$.
Therefore $\Omega^{\mathrm{ext}}_3 =\Set{\bsr\in\bbR^8:
\norm{\bsr}=1, \bsr\star\bsr=\bsr}$.
\end{proof}

\begin{remark}
There is another characterization of $\Omega_3$, which is described
by
\begin{eqnarray}
\Omega_3=\Set{\bsr\in\bbR^8:\norm{\bsr}\leqslant\tfrac12}\cup
\Set{\bsr\in\bbR^8:\norm{\bsr}\in\Br{\tfrac12,1},\arccos\Pa{\tfrac{\Inner{\bsr}{\bsr\star\bsr}}{\norm{\bsr}^3}}+3\arccos\Pa{\tfrac1{2\norm{\bsr}}}\leqslant\pi}.
\end{eqnarray}
Indeed, for a $3\times 3$ generic Hermitian matrix of fixed-trace
one $\rho(\bsr)=\frac13(\I+\sqrt{3}\bsr\cdot\bG)$, where
$\bsr\neq\zero$, its all eigenvalues are given by
\begin{eqnarray*}
\lambda_1(\rho(\bsr))&=& \frac13+\frac23\cos\theta,\\
\lambda_2(\rho(\bsr))&=& \frac13+\frac23\cos\Pa{\theta-\frac{2\pi}3},\\
\lambda_3(\rho(\bsr))&=&
\frac13+\frac23\cos\Pa{\theta+\frac{2\pi}3},
\end{eqnarray*}
where
$\theta=\theta(\bsr)=\frac13\arccos\Pa{\norm{\bsr}^{-3}\Inner{\bsr}{\bsr\star\bsr}}$.
Note that
$$
\norm{\bsr}^{-3}\Inner{\bsr}{\bsr\star\bsr}\in[-1,1]\quad\text{and}\quad
\arccos(\norm{\bsr}^{-3}\Inner{\bsr}{\bsr\star\bsr})\in[0,\pi].
$$
Then $\theta\in[0,\tfrac\pi3]$. Based on this observation, we see
that
$\lambda_1(\rho(\bsr))\geqslant\lambda_2(\rho(\bsr))\geqslant\lambda_3(\rho(\bsr))$.
In order to get the positivity of $\rho(\bsr)$, it suffices to
characterize $\lambda_3(\rho(\bsr))\geqslant0$. That is,
$2\norm{\bsr}\cos(\theta+\tfrac{2\pi}3)\geqslant-1$. Clearly if
$\norm{\bsr}\leqslant\tfrac12$, then
$\lambda_3(\rho(\bsr))\geqslant0$. We can assume that
$\norm{\bsr}\geqslant\tfrac12$, then
$$
\cos\Pa{\theta-\frac\pi3}=\cos\Pa{\frac\pi3-\theta}\leqslant\frac1{2\norm{\bsr}}\leqslant1.
$$
In summary, we obtain that
\begin{eqnarray*}
\frac\pi3-\theta\geqslant\arccos\Pa{\frac1{2\norm{\bsr}}}\Longleftrightarrow\arccos\Pa{\tfrac{\Inner{\bsr}{\bsr\star\bsr}}{\norm{\bsr}^3}}+3\arccos\Pa{\tfrac1{2\norm{\bsr}}}\leqslant\pi.
\end{eqnarray*}
In what follows, we show that
\begin{eqnarray}
&&\Set{\bsr\in\bbR^8:\norm{\bsr}\leqslant\tfrac12}\cup
\Set{\bsr\in\bbR^8:\norm{\bsr}\in\Br{\tfrac12,1},\arccos\Pa{\tfrac{\Inner{\bsr}{\bsr\star\bsr}}{\norm{\bsr}^3}}+3\arccos\Pa{\tfrac1{2\norm{\bsr}}}\leqslant\pi}\notag\\
&&=\Set{\bsr\in\bbR^8:\norm{\bsr}\leqslant1,1+2\Inner{\bsr}{\bsr\star\bsr}\geqslant3\Inner{\bsr}{\bsr}}.
\end{eqnarray}
Because
$\abs{2\norm{\bsr}\cos(\theta\pm\tfrac{2\pi}3)}\leqslant2\norm{\bsr}$,
which means that
$2\norm{\bsr}\cos(\theta\pm\tfrac{2\pi}3)\geqslant-2\norm{\bsr}$. Then
$$
\lambda_3(\rho)\geqslant\frac{1-2\abs{\bsr}}3\geqslant0
\quad\text{if }\norm{\bsr}\leqslant\frac12.
$$
In what follows, we assume that $\norm{\bsr}\geqslant\frac12$. Since
$\Tr{\rho(\bsr)^2}=\frac{1+2\norm{\bsr}^2}3$ and $\rho(\bsr)$ is a
legal state iff $\bsr\in\Omega_3$, it follows that
$\Tr{\rho(\bsr)^2}\leqslant1$, i.e., $\norm{\bsr}\leqslant1$.

With the above assumption, $\frac12\leqslant\norm{\bsr}\leqslant1$,
and thus $\frac1{2\norm{\bsr}}\leqslant1$. Let
$x=\frac{\Inner{\bsr}{\bsr\star\bsr}}{\norm{\bsr}^3}\in[-1,1]$ and
$y=\frac1{2\norm{\bsr}}\in[\tfrac12,1]$. Clearly we always have
$\arccos y\in[0,\frac\pi3]$. We separate it into two cases:
\begin{enumerate}[(i)]
\item If $x\geqslant0$, then $0\leqslant\arccos x\leqslant\frac\pi2$. Thus $\arccos x+3\arccos y\leqslant
\pi$ will give rise to $0\leqslant\arccos
x\leqslant\min\Set{\pi-3\arccos y,\frac\pi2}$, i.e.,
$$
x\geqslant \max\Set{-\cos(3\arccos y),0}=\max\Set{3y-4y^3,0}.
$$
\item If $x\leqslant0$, then $\frac\pi2\leqslant\arccos x\leqslant\pi$. Thus $\arccos x+3\arccos y\leqslant
\pi$ will give rise to $3\arccos y\leqslant \pi-\arccos
x\leqslant\frac\pi2$, i.e.,
$$
\cos(3\arccos y)=4y^3-3y\geqslant -x.
$$
\end{enumerate}
In summary, we always have $x\geqslant 3y-4y^3$. that is
$$
\tfrac{\Inner{\bsr}{\bsr\star\bsr}}{\abs{\bsr}^3}\geqslant
3\tfrac1{2\abs{\bsr}}-4\Pa{\tfrac1{2\abs{\bsr}}}^3\Longleftrightarrow1+2\Inner{\bsr}{\bsr\star\bsr}\geqslant
3\Inner{\bsr}{\bsr}.
$$
In fact, if $\norm{\bsr}\leqslant\frac12$, then
$$
\abs{\Inner{\bsr}{\bsr\star\bsr}}\leqslant\norm{\bsr}^3\leqslant\frac18,
$$
implying that $\Inner{\bsr}{\bsr\star\bsr}\geqslant-\frac18$. Thus
\begin{eqnarray*}
1+2\Inner{\bsr}{\bsr\star\bsr}\geqslant 1-\frac14=\frac34 =
3\Pa{\frac12}^2\geqslant 3\norm{\bsr}^2=3\Inner{\bsr}{\bsr}.
\end{eqnarray*}
\end{remark}

%---------------------------------------------------------%
\subsection{Parametrization of $\Omega^{\text{ext}}_3$}
%---------------------------------------------------------%

As a submanifold of the manifold $\Omega_3$, $\Omega^{\text{ext}}_3$
is $4$-dimensional. Now for
$\bsr=(r_1,\ldots,r_8)\in\Omega^{\text{ext}}_3$, we see that
$$
\bsr\star\bsr=\bsr\quad\text{and}\quad\norm{\bsr}=1,
$$
that is,
\begin{eqnarray*}
\Innerm{\bsr}{\bsD_k}{\bsr}&=&r_k,\quad k=1,\ldots,8;\\
\sum^8_{k=1}r^2_k&=&1.
\end{eqnarray*}
The above constraints are reduced into the following forms:
\begin{eqnarray}
0&=&\sqrt{3}(r_4r_6+r_5r_7)+r_1(2r_8-1),\label{eq:1}\\
0&=&\sqrt{3}(r_5r_6-r_4r_7)+r_2(2r_8-1),\label{eq:2}\\
0&=&\sqrt{3}(r^2_4+r^2_5-r^2_6-r^2_7)+2r_3(2r_8-1),\label{eq:3}\\
0&=&\sqrt{3}(r_1r_6-r_2r_7)+r_4(\sqrt{3}r_3-r_8-1),\notag\\
0&=&\sqrt{3}(r_2r_6+r_1r_7)+r_5(\sqrt{3}r_3-r_8-1),\notag\\
0&=&\sqrt{3}(r_1r_4+r_2r_5)-r_6(\sqrt{3}r_3+r_8+1),\notag\\
0&=&\sqrt{3}(r_2r_4-r_1r_5)+r_7(\sqrt{3}r_3+r_8+1),\notag\\
0&=&2(r^2_1+r^2_2+r^2_3) - (r^2_4+r^2_5+r^2_6+r^2_7) -2r_8(r_8+1),\label{eq:4}\\
1&=&r^2_1+r^2_2+r^2_3+r^2_4+r^2_5+r^2_6+r^2_7+r^2_8.\label{eq:5}
\end{eqnarray}
Let $R^2=r^2_4+r^2_5+r^2_6+r^2_7\in[0,1]$, where $R\geqslant0$. Then
$r^2_1+r^2_2+r^2_3+r^2_8=1-R^2$ by Eq.~\eqref{eq:5}. From
Eq.~\eqref{eq:4}, it follows that
$2(1-R^2-r^2_8)-R^2-2r_8(r_8+1)=0$, i.e.,
\begin{eqnarray*}
4r^2_8+2r_8+(3R^2-2)=0.
\end{eqnarray*}
As the quadratic equation of argument $r_8$, its roots must be real.
This indicates that its discriminant $\Delta=2^2-4\cdot
4(3R^2-2)\geqslant0$. That is, $R^2\leqslant\frac34$. At this time,
the solution is given by
\begin{eqnarray}
r_8(\epsilon)=\frac14\Pa{-1+\epsilon\sqrt{3(3-4R^2)}},\quad \forall
R\in[0, \tfrac{\sqrt{3}}2]\text{ and }\epsilon\in\set{\pm1}.
\end{eqnarray}
\begin{itemize}
\item If $R=0$, i.e., $(r_4,r_5,r_6,r_7)=(0,0,0,0)$, then
$r_8(\epsilon)=\frac{3\epsilon-1}4$. Based on this, we have
\begin{eqnarray*}
0&=& r_1(2r_8-1) = r_2(2r_8-1) = r_3(2r_8-1)\\
1&=& r^2_1+r^2_2+r^2_3+r^2_8
\end{eqnarray*}
If $\epsilon=-1$, $r_8(\epsilon)=-1$, and $(r_1,r_2,r_3)=(0,0,0)$;
if $\epsilon=1$, $r_8(\epsilon)=\frac12$, and
$r^2_1+r^2_2+r^2_3=\frac34$.
\item If $R\in (0,\tfrac{\sqrt{3}}2]$, then $2r_8(\epsilon)-1<0$. Solving this
group of equations, we get that
\begin{eqnarray}
r_1(\epsilon) &=& \frac{\sqrt{3}+\epsilon\sqrt{3-4R^2}}{2R^2}(r_4
r_6+r_5
r_7),\\
r_2(\epsilon) &=& \frac{\sqrt{3}+\epsilon\sqrt{3-4R^2}}{2R^2}(r_5
r_6-r_4
r_7),\\
r_3(\epsilon) &=& \frac{\sqrt{3}+\epsilon\sqrt{3-4R^2}}{4R^2}(r^2_4
+ r^2_5-r^2_6-
r^2_7),\\
r_8(\epsilon) &=& \frac{-1+\epsilon\sqrt{3(3-4R^2)}}4.
\end{eqnarray}
\end{itemize}
Thus, as the free variables, $(r_4,r_5,r_6,r_7)$ is the suitable
choice for the parametrization of $\Omega^{\text{ext}}_3$.

In summary, for a generic $\bsr\in \Omega^{\text{ext}}_3$, it can be
classified into the following three categories $\cI_1,\cI_2$, and
$\cI_3$:
\begin{eqnarray}
\cI_1&=&\set{\Pa{0,0,0,0,0,0,0,-1}},\label{eq:I1}\\
\cI_2&=&\Set{\Pa{r_1,r_2,r_3,0,0,0,0,\tfrac12}: r^2_1+r^2_2+r^2_3=\frac34},\label{eq:I2}\\
\cI_3&=&\Set{\Pa{r_1(\epsilon),r_2(\epsilon),r_3(\epsilon),r_4,r_5,r_6,r_7,r_8(\epsilon)}:
0<\sum^7_{i=4}r^2_i\leqslant \frac34 \text{ for
}\epsilon\in\set{\pm1}}.\label{eq:I3}
\end{eqnarray}
That is,
\begin{eqnarray}\label{eq:partition}
\Omega^{\text{ext}}_3=\cI_1\cup\cI_2\cup\cI_3.
\end{eqnarray}
For a convenient usage, denote
\begin{eqnarray}
\cJ_R(\epsilon)=\Set{\bsr(\epsilon)\in\cI_3:\sum^7_{i=4}r^2_i=R^2}
\end{eqnarray}
for $R\in(0,\tfrac{\sqrt{3}}2]$ and $\epsilon\in\set{\pm1}$. Let
$\cJ_R=\cup_{\epsilon\in\set{\pm1}}\cJ_R(\epsilon)$. Then
\begin{eqnarray}
\cI_3=\bigcup_{R\in(0,\sqrt{3}/2]}\cJ_R.
\end{eqnarray}

%----------------------------------------------------%
\subsection{Examples with analytical computations}
%----------------------------------------------------%

\begin{exam}
For given two qutrit observables
$$
\bsA=\Pa{\begin{array}{ccc}
                                      -1 & 0 & 0 \\
                                      0 & 0 & 0 \\
                                      0 & 0 & 1
                                    \end{array}
}=-\frac12\bsG_3-\frac{\sqrt{3}}2\bsG_8\quad\text{and}\quad\bsB=\Pa{\begin{array}{ccc}
                   0 & 0 & 0 \\
                   0 & 0 & \mathrm{i} \\
                   0 & -\mathrm{i} & 0
                 \end{array}
}=-\bsG_7,
$$
we have that
$$
\min_{\rho\in\rD(\bbC^3)}[\var_\rho(\bsA)+\var_\rho(\bsB)]=0.
$$
Indeed, we rewrite $\bsA=a_0\I+\bsa\cdot\bG$ and
$\bsB=b_0\I+\bsb\cdot\bG$, where $a_0=b_0=0$ and
\begin{eqnarray*}
\bsa =\Pa{0,0,-\tfrac12,0,0,0,0,-\tfrac{\sqrt{3}}2},\quad \bsb =
\Pa{0,0,0,0,0,0,-1,0}.
\end{eqnarray*}
Then $\bsT_{\bsa,\bsb} = -\bsD_8 - 2\bsO_{\bsa,\bsb}$, which is
equal to
\begin{eqnarray}
\bsT_{\bsa,\bsb} =\Pa{
\begin{array}{cccccccc}
 -1 & 0 & 0 & 0 & 0 & 0 & 0 & 0 \\
 0 & -1 & 0 & 0 & 0 & 0 & 0 & 0 \\
 0 & 0 & -\frac{3}{2} & 0 & 0 & 0 & 0 & -\frac{\sqrt{3}}{2} \\
 0 & 0 & 0 & \frac{1}{2} & 0 & 0 & 0 & 0 \\
 0 & 0 & 0 & 0 & \frac{1}{2} & 0 & 0 & 0 \\
 0 & 0 & 0 & 0 & 0 & \frac{1}{2} & 0 & 0 \\
 0 & 0 & 0 & 0 & 0 & 0 & -\frac{3}{2} & 0 \\
 0 & 0 & -\frac{\sqrt{3}}{2} & 0 & 0 & 0 & 0 & -\frac{1}{2} \\
\end{array}
}.
\end{eqnarray}
Now
\begin{eqnarray}
\Innerm{\bsr}{\bsT_{\bsa,\bsb}}{\bsr} =\frac12\Pa{-2 r_1^2-2 r_2^2-3
r_3^2+r_4^2+r_5^2+r_6^2-3 r_7^2-r_8^2-2 \sqrt{3} r_3 r_8}.
\end{eqnarray}
\begin{itemize}
\item $\min_{\bsr\in\cI_1}\Innerm{\bsr}{\bsT_{\bsa,\bsb}}{\bsr}
=-\frac12$.
\item For $\bsr\in\cI_2$, we see that $\Innerm{\bsr}{\bsT_{\bsa,\bsb}}{\bsr} =\frac{-8(r^2_1+r^2_2+r^2_3)-4r_3^2-4 \sqrt{3} r_3-1}8=\frac{-4r^2_3-4\sqrt{3}r_3-7}8$. By using spherical coordinates, we let
$r_3=\frac{\sqrt{3}}2\cos\theta$, then we have
$\min_{\bsr\in\cI_2}\Innerm{\bsr}{\bsT_{\bsa,\bsb}}{\bsr} = -2$.
\item For $\bsr\in \cJ_R:=\cup_{\epsilon\in\set{\pm1}}\cJ_R(\epsilon)$, we see that
$\sum^3_{k=1}r^2_k(\epsilon)+r^2_8(\epsilon)=1-R^2$, where
$\sum^7_{k=1}r^2_k=R^2$, then
\begin{eqnarray*}
\Innerm{\bsr}{\bsT_{\bsa,\bsb}}{\bsr} = \frac12\Pa{-r_3^2-4
r_7^2+r_8^2-2 \sqrt{3} r_3 r_8+3 R^2-2}.
\end{eqnarray*}
We construct Lagrange multiplier function as follows:
$$
L(r_4,r_5,r_6,r_7,\lambda)
=\Innerm{\bsr}{\bsT_{\bsa,\bsb}}{\bsr}+\lambda\Pa{\sum^7_{k=4}
r^2_k-R^2}.
$$
Then $\frac{\partial L}{\partial r_4}=\frac{\partial L}{\partial
r_5}=0$ means that $r_4=r_5=0$; and $\frac{\partial L}{\partial
r_6}=\frac{\partial L}{\partial r_7}=0$ means that
$$
\lambda=3,r_6=0, r_7=\pm R\quad\text{or}\quad \lambda=1,r_6=\pm
R,r_7=0.
$$
For $\lambda=3$ and $(r_4,r_5,r_6,r_7)=(0,0,0,\pm R)$ and
$$
(r_1(\epsilon),r_2(\epsilon),r_3(\epsilon),r_8(\epsilon))=\Pa{0,0,\frac{-\sqrt{3}-\epsilon\sqrt{3-4R^2}}4,\frac{-1+\epsilon\sqrt{3(3-4R^2)}}4},
$$
and thus
\begin{eqnarray*}
\Innerm{\bsr(\epsilon)}{\bsT_{\bsa,\bsb}}{\bsr(\epsilon)} =-
\frac{3R^2+1}2\quad\forall \epsilon\in\set{\pm1}.
\end{eqnarray*}
We get that
\begin{eqnarray*}
\min_{R\in(0,\sqrt{3}/2]}\min_{\bsr\in\cJ_R}\Innerm{\bsr}{\bsT_{\bsa,\bsb}}{\bsr}=
- \max_{R\in(0,\sqrt{3}/2]}\frac{3R^2+1}2 = -\frac{13}8
\end{eqnarray*}
which is attained at $R=\frac{\sqrt{3}}2$. That is,
$\min_{\bsr\in\cI_3}\Innerm{\bsr}{\bsT_{\bsa,\bsb}}{\bsr}=-\frac{13}8$.
\end{itemize}
Therefore
\begin{eqnarray*}
\ell_{\bsa,\bsb} &=&
\min\Set{\min_{\bsr\in\cI_1}\Innerm{\bsr}{\bsT_{\bsa,\bsb}}{\bsr},\min_{\bsr\in\cI_2}\Innerm{\bsr}{\bsT_{\bsa,\bsb}}{\bsr},\min_{\bsr\in\cI_3}\Innerm{\bsr}{\bsT_{\bsa,\bsb}}{\bsr}}\\
&=& \min\Set{-\frac12,-2,-\frac{13}8} = -2
\end{eqnarray*}
Now $m_{\bsA,\bsB}=\frac23(1+1-2)=0$ is attained at the qutrit state
\begin{eqnarray*}
\rho(\bsr)=\frac13(\I+\sqrt{3}\bsr\cdot\bG)=\Pa{
\begin{array}{ccc}
 1 & 0 & 0 \\
 0 & 0 & 0 \\
 0 & 0 & 0 \\
\end{array}
},
\end{eqnarray*}
where
$$
\bsr = \Pa{0,0,\frac{\sqrt{3}}2,0,0,0,0,\frac12}.
$$
\end{exam}

\begin{exam}\label{exam:gigj}
For given two qutrit observables
$$
\bsA=\bsG_4 \quad\text{and}\quad \bsB=\bsG_6,
$$
we have that \cite{Kurzynski2016}
$$
\min_{\rho\in\rD(\bbC^3)}[\var_\rho(\bsA)+\var_\rho(\bsB)]=\frac7{16}.
$$
Indeed, we rewrite $\bsA=\bsa\cdot\bG$ and $\bsB=\bsb\cdot\bG$,
where
\begin{eqnarray*}
\bsa =\Pa{0,0,0,1,0,0,0,0},\quad \bsb = \Pa{0,0,0,0,0,1,0,0}.
\end{eqnarray*}
Then
\begin{eqnarray*}
\bsT_{\bsa,\bsb}=-\bsD_8-2\bsO_{\bsa,\bsb}
=\diag\Pa{-1,-1,-1,-\frac32,\frac12,-\frac32,\frac12,1}.
\end{eqnarray*}
\begin{itemize}
\item $\min_{\bsr\in\cI_1}\Innerm{\bsr}{\bsT_{\bsa,\bsb}}{\bsr}=1$
\item $\min_{\bsr\in\cI_2}\Innerm{\bsr}{\bsT_{\bsa,\bsb}}{\bsr}=-\frac12$. Indeed, for $\bsr\in\cI_2$, we see that
$\Innerm{\bsr}{\bsT_{\bsa,\bsb}}{\bsr}=\frac14-(r^2_1+r^2_2+r^2_3)=-\frac12$.
\item
$\min_{\bsr\in\cI_3}\Innerm{\bsr}{\bsT_{\bsa,\bsb}}{\bsr}=-\frac{43}{32}$.
In fact, for $\bsr\in\cJ_R$,
\begin{eqnarray*}
\Innerm{\bsr}{\bsT_{\bsa,\bsb}}{\bsr}&=&\frac12\Br{4\Pa{r^2_8-r^2_4-r^2_6}-2\Pa{r^2_1+r^2_2+r^2_3+r^2_8}+\Pa{r^2_4+r^2_5+r^2_6+r^2_7}}\\
&=&2\Pa{r^2_8-r^2_4-r^2_6}+\frac{-2(1-R^2)+R^2}2 \\
&=& 2\Pa{r^2_8-r^2_4-r^2_6}+\Pa{\frac32R^2-1}.
\end{eqnarray*}
Note that $r_8=\frac{-1+\epsilon\sqrt{3(3-4R^2)}}4$. We have
$\Innerm{\bsr}{\bsT_{\bsa,\bsb}}{\bsr}=
-2\Pa{r^2_4+r^2_6}+\frac{1-\epsilon\sqrt{3(3-4R^2)}}4$, and thus
\begin{eqnarray*}
\min_{\bsr\in\cJ_R(\epsilon)}\Innerm{\bsr}{\bsT_{\bsa,\bsb}}{\bsr}&=&
-2\max_{\bsr\in\cJ_R(\epsilon)}\Pa{r^2_4+r^2_6}+\frac{1-\epsilon\sqrt{3(3-4R^2)}}4\\
&=&-2R^2+\frac{1-\epsilon\sqrt{3(3-4R^2)}}4,
\end{eqnarray*}
where in the last equality, we used the fact that
$r^2_4+r^2_6+r^2_5+r^2_7=R^2$ and
$\max_{\bsr\in\cJ_R(\epsilon)}\Pa{r^2_4+r^2_6}=R^2$ only if
$(r_5,r_7)=(0,0)$. Then
\begin{eqnarray*}
&&\min_{\bsr\in\cI_3}\Innerm{\bsr}{\bsT_{\bsa,\bsb}}{\bsr}=\min_{R\in(0,\sqrt{3}/2]}\min_{\bsr\in\cJ_R}\Innerm{\bsr}{\bsT_{\bsa,\bsb}}{\bsr}\\
&&=\min_{R\in(0,\sqrt{3}/2]}\Br{-2R^2+\frac{1-\sqrt{3(3-4R^2)}}4}=-\frac{43}{32},
\end{eqnarray*}
which is attained at $R=\frac{3\sqrt{5}}{8}$.
\end{itemize}
From the above discussion, we get that
\begin{eqnarray*}
\ell_{\bsa,\bsb}&=&
\min\Set{\min_{\bsr\in\cI_1}\Innerm{\bsr}{\bsT_{\bsa,\bsb}}{\bsr},\min_{\bsr\in\cI_2}\Innerm{\bsr}{\bsT_{\bsa,\bsb}}{\bsr},\min_{\bsr\in\cI_3}\Innerm{\bsr}{\bsT_{\bsa,\bsb}}{\bsr}}
\\
&=&\min\Set{1,-\frac12,-\frac{43}{32}}=-\frac{43}{32},
\end{eqnarray*}
from which we obtain that
$$
m_{\bsA,\bsB} =
\frac23\Pa{\abs{\bsa}^2+\abs{\bsb}^2+\ell_{\bsa,\bsb}} = \frac7{16}.
$$
In this situation, we present the specific form of qutrit state
saturating the lower bound in the additive uncertainty relation
$\var_\rho(\bsG_4)+\var_\rho(\bsG_6)\geqslant\frac7{16}$:
\begin{eqnarray}
\rho(\bsr) = \frac13(\I+\sqrt{3}\bsr\cdot\bG)=\frac1{16}\Pa{\begin{array}{ccc}
                                                              5(1+2\cos(2t)) & 5\sin(2t) & 2\sqrt{15}\cos t \\
                                                              5\sin(2t) & 5(1-2\cos(2t)) & 2\sqrt{15}\sin t \\
                                                              2\sqrt{15}\cos t & 2\sqrt{15}\sin t &
                                                              6
                                                            \end{array}
},
\end{eqnarray}
where
$$
\bsr=\Pa{\frac{5\sqrt{3}}{16} \sin(2t),0,\frac{5\sqrt{3}}{8}\cos(2
t),\frac{3\sqrt{5}}{8}\cos t,0,\frac{3\sqrt{5}}{8}\sin
t,0,-\frac{1}{16}}.
$$
\end{exam}

\begin{exam}
Choose the angular momentum operators
\cite{Gao2021,Toth2022,Chiew2022}
$$
\bsL_x = \frac1{\sqrt{2}}\Pa{\begin{array}{ccc}
                               0 & 1 & 0 \\
                               1 & 0 & 1 \\
                               0 & 1 & 0
                             \end{array}
},\quad \bsL_y =\frac1{\sqrt{2}}\Pa{\begin{array}{ccc}
                               0 & -\mathrm{i} & 0 \\
                               \mathrm{i} & 0 & -\mathrm{i} \\
                               0 & \mathrm{i} & 0
                             \end{array}
},\quad \bsL_z=\Pa{\begin{array}{ccc}
                     1 & 0 & 0 \\
                     0 & 0 & 0 \\
                     0 & 0 & -1
                   \end{array}
}.
$$
Note that
$\bsL_x=\frac{\bsG_1+\bsG_6}{\sqrt{2}},\bsL_y=\frac{\bsG_2+\bsG_7}{\sqrt{2}}$,
and $\bsL_z=\frac{\bsG_3+\sqrt{3}\bsG_8}2$. We have that
\begin{eqnarray*}
\min_{\rho\in\density{\bbC^3}}\Br{\var_\rho(\bsL_x)+\var_\rho(\bsL_y)+\var_\rho(\bsL_z)}
=1.
\end{eqnarray*}
Indeed, we can rewrite them as
$$
\bsL_x=\bsa_x\cdot\bG,\quad \bsL_y = \bsa_y\cdot\bG,\quad
\bsL_z=\bsa_z\cdot\bG,
$$
where
\begin{eqnarray*}
\bsa_x &=& \Pa{\frac1{\sqrt{2}},0,0,0,0,\frac1{\sqrt{2}},0, 0},\\
\bsa_y &=& \Pa{0,\frac1{\sqrt{2}},0,0,0,0,\frac1{\sqrt{2}},0},\\
\bsa_z &=& \Pa{0,0,\frac12,0, 0, 0, 0, \frac{\sqrt{3}}2}.
\end{eqnarray*}
Then $\bsT_{x,y,z}=-2(\proj{\bsa_x}+\proj{\bsa_y}+\proj{\bsa_z})$,
that is,
\begin{eqnarray*}
\bsT_{x,y,z}=\Pa{
\begin{array}{cccccccc}
 -1 & 0 & 0 & 0 & 0 & -1 & 0 & 0 \\
 0 & -1 & 0 & 0 & 0 & 0 & -1 & 0 \\
 0 & 0 & -\frac{1}{2} & 0 & 0 & 0 & 0 & -\frac{\sqrt{3}}{2} \\
 0 & 0 & 0 & 0 & 0 & 0 & 0 & 0 \\
 0 & 0 & 0 & 0 & 0 & 0 & 0 & 0 \\
 -1 & 0 & 0 & 0 & 0 & -1 & 0 & 0 \\
 0 & -1 & 0 & 0 & 0 & 0 & -1 & 0 \\
 0 & 0 & -\frac{\sqrt{3}}{2} & 0 & 0 & 0 & 0 & -\frac{3}{2} \\
\end{array}
}.
\end{eqnarray*}
Now
\begin{eqnarray}
\Innerm{\bsr}{\bsT_{x,y,z}}{\bsr}
=-\frac12\Br{2(r_1+r_6)^2+2(r_2+r_7)^2+(r_3+\sqrt{3}r_8)^2}.
\end{eqnarray}
Thus
\begin{itemize}
\item
$\min_{\bsr\in\cI_1}\Innerm{\bsr}{\bsT_{x,y,z}}{\bsr}=-\frac32$.
\item $\min_{\bsr\in\cI_2}\Innerm{\bsr}{\bsT_{x,y,z}}{\bsr}=-\frac32$.
Indeed, for $\bsr\in\cI_2$,
$\Innerm{\bsr}{\bsT_{x,y,z}}{\bsr}=-\Pa{r^2_1+r^2_2+r^2_3}+\frac{4r^2_3-4\sqrt{3}r_3-3}8=\frac{4r^2_3-4\sqrt{3}r_3-9}8$.
Again, using spherical coordinate, let
$r_3=\frac{\sqrt{3}}2\cos\theta$, and we see that
$$
\min_{\bsr\in\cI_2}\Innerm{\bsr}{\bsT_{x,y,z}}{\bsr}=-\frac32,
$$
which is attained at $r_3=\frac{\sqrt{3}}2$.
\item Let $\alpha=\frac{\sqrt{3}+\sqrt{3-4 R^2}\epsilon}{4R^2}$ and $\beta=\frac{-1+\epsilon\sqrt{3(3-4 R^2)}}4$. For $\bsr\in\cI_3$,
\begin{eqnarray*}
\Innerm{\bsr}{\bsT_{x,y,z}}{\bsr}&=&\frac{r_3^2}{2}-\sqrt{3} r_8
r_3-r_6^2-r_7^2-2 r_1 r_6-2 r_2 r_7-\frac{r_8^2}{2} -
\Pa{r^2_1+r^2_2+r^2_3+r^2_8}\\
&=&\frac{r_3^2}{2}-\sqrt{3} r_8 r_3-r_6^2-r_7^2-2 r_1 r_6-2 r_2
r_7-\frac{r_8^2}{2}+R^2-1.
\end{eqnarray*}
For the Lagrange multiplier function
$L=\Innerm{\bsr}{\bsT_{x,y,z}}{\bsr}-\lambda\Pa{\sum^7_{k=4}r^2_k-R^2}$,
the vanishing gradient $\nabla L=\zero$ means that
\begin{eqnarray*}
2\lambda r_4 &=& -2\sqrt{3}\alpha\beta r_4-4\alpha r^2_6+4\alpha
r^2_7+2\alpha^2r_4(R^2-2 r_6^2-2 r_7^2),\\
2\lambda r_5 &=& 4\alpha^2r_4r_5,\\
2\lambda r_6 &=& 0,\\
2\lambda r_7 &=& 0.
\end{eqnarray*}
\begin{itemize}
\item If $\lambda\neq0$, then $r_6=r_7=0$. Thus $(2\lambda-1)r_4=0$
and $(\lambda-2\alpha^2r_4)r_5=0$. So $r_4\neq0$ by
$r^2_4+r^2_5=R^2$. Otherwise $r_4=0$ leads $r_5=0$, a contradiction.
Then $\lambda=\frac12$. If $r_5\neq0$, then $r_4=\frac1{4\alpha^2}$
and $r^2_5=R^2-r^2_4=R^2-\frac1{16\alpha^4}$; if $r_5=0$, then
$r_4=\pm R$. Thus $\Innerm{\bsr}{\bsT_{x,y,z}}{\bsr}=-\frac{\beta
^2}{2}+\frac{1}{2}\alpha^2(r_4^2+r_5^2)^2-\sqrt{3}\alpha\beta
(r_4^2+r_5^2)+R^2-1$, and
\begin{eqnarray*}
\inf_{\bsr\in\cI_3}\Innerm{\bsr}{\bsT_{x,y,z}}{\bsr} =
\inf_{0<R\leqslant \sqrt{3}/2}\frac{4R^2-3}2=-\frac32.
\end{eqnarray*}
\item If $\lambda=0$, then $r_4r_5=0$. For $r_4=0$, then it must
have $r_6=r_7=0$, thus $r_5=\pm R$. For $r_5=0$,
\begin{eqnarray*}
r_4=\frac{(\sqrt{3}-\epsilon\sqrt{3-4R^2})(r^2_6-r^2_7)}{\Pa{\frac{\sqrt{3}-\epsilon\sqrt{3-4R^2}}2}^2-(r^2_6+r^2_7)}
= \frac{\frac1\alpha(r^2_6-r^2_7)}{\frac1{4\alpha^2}-(r^2_6+r^2_7)},
\end{eqnarray*}
together with $r^2_6+r^2_7=R^2-r^2_4$, we get that
$$
r^2_6-r^2_7=\alpha r_4\Pa{r_4^2+\frac{1}{4\alpha^2}-R^2}
$$
implying that
\begin{eqnarray*}
r^2_6 &=& \frac{4\alpha^2 r^3_4-4\alpha r^2_4-4\alpha^2R^2 r_4+r_4+4\alpha R^2}{8\alpha},\\
r^2_7 &=& \frac{-4\alpha^2 r^3_4-4\alpha r^2_4+4\alpha^2R^2
r_4-r_4+4\alpha R^2}{8\alpha}.
\end{eqnarray*}
\end{itemize}
Then
\begin{eqnarray*}
\min_{\bsr\in\cJ_R}\Innerm{\bsr}{\bsT_{x,y,z}}{\bsr}
&=&\min_{r_4\in[-R,R]} \frac{4r^2_4R^4+r^4_4 (-\sqrt{9-12
R^2}\epsilon +2R^2-3)+R^4(\sqrt{9-12 R^2}\epsilon
-2R^2-3)}{4R^4}\\
&=&-R^2,
\end{eqnarray*}
if $r^2_4=\frac{2R^4}{3-2R^2+\epsilon\sqrt{3(3-4R^2)}}$, which leads
to the following result
$$
\min_{\bsr\in\cI_3}\Innerm{\bsr}{\bsT_{x,y,z}}{\bsr}=-\max_{0<R\leqslant\sqrt{3}/2}R^2
= -\frac34.
$$
\end{itemize}
Therefore $\ell_{x,y,z}=\min\set{-\frac32,-\frac34}=-\frac32$. Then
$m_{\bsL_x,\bsL_y,\bsL_z}=\frac23\Pa{1+1+1-\frac32}=1$.
\end{exam}

\begin{prop}
For Gell-Mann matrices, denote
\begin{eqnarray}
m_{ij}:=\min_{\rho\in\density{\bbC^3}}\Br{\var_\rho(\bsG_i) +
\var_\rho(\bsG_j)},\quad (1\leqslant i<j\leqslant8).
\end{eqnarray}
It holds that all non-vanishing ones of $\binom{8}{2}=28$ constants
are given below:
\begin{eqnarray*}
m_{14} =m_{15} =m_{16} =m_{17} =m_{24} =m_{25} =m_{26}
=m_{27}=m_{46} =m_{47} =m_{56} =m_{57} =\frac7{16}.
\end{eqnarray*}
\end{prop}

\begin{proof}
The proof goes similarly for the calculation performed in
Example~\ref{exam:gigj}.
\end{proof}

\begin{exam}
For given two qutrit observables
$$
\bsA=\Pa{\begin{array}{ccc}
                                      -1 & 0 & 0 \\
                                      0 & 0 & 0 \\
                                      0 & 0 & 1
                                    \end{array}
}\quad\text{and}\quad\bsB=\Pa{\begin{array}{ccc}
                   0 & 1 & 0 \\
                   1 & 0 & \mathrm{i} \\
                   0 & -\mathrm{i} & 0
                 \end{array}
},
$$
we have that
$$
\min_{\rho\in\rD(\bbC^3)}[\var_\rho(\bsA)+\var_\rho(\bsB)]=\frac{15}{32}.
$$
Indeed, we rewrite $\bsA=\bsa\cdot\bG$ and $\bsB=\bsb\cdot\bG$ and
\begin{eqnarray*}
\bsa =\Pa{0,0,-\tfrac12,0,0,0,0,-\tfrac{\sqrt{3}}2},\quad \bsb =
\Pa{1,0,0,0,0,0,-1,0}.
\end{eqnarray*}
Then $\bsT_{\bsa,\bsb} = -\sqrt{3}\bsD_5 - 2\bsO_{\bsa,\bsb}$, which
is equal to
\begin{eqnarray*}
\bsT_{\bsa,\bsb} =\Pa{
\begin{array}{cccccccc}
 -2 & 0 & 0 & 0 & 0 & 0 & \frac{1}{2} & 0 \\
 0 & 0 & 0 & 0 & 0 & -\frac{3}{2} & 0 & 0 \\
 0 & 0 & -\frac{1}{2} & 0 & -\frac{3}{2} & 0 & 0 & -\frac{\sqrt{3}}{2} \\
 0 & 0 & 0 & 0 & 0 & 0 & 0 & 0 \\
 0 & 0 & -\frac{3}{2} & 0 & 0 & 0 & 0 & \frac{\sqrt{3}}{2} \\
 0 & -\frac{3}{2} & 0 & 0 & 0 & 0 & 0 & 0 \\
 \frac{1}{2} & 0 & 0 & 0 & 0 & 0 & -2 & 0 \\
 0 & 0 & -\frac{\sqrt{3}}{2} & 0 & \frac{\sqrt{3}}{2} & 0 & 0 & -\frac{3}{2} \\
\end{array}
}.
\end{eqnarray*}
Now the spectrum of $\bsT_{\bsa,\bsb}$ is given by
$\set{-\tfrac52,-2,-\sqrt{3},\Pa{-\tfrac32}_{(2)},0,\tfrac32,\sqrt{3}}$
and
\begin{eqnarray*}
\Innerm{\bsr}{\bsT_{\bsa,\bsb}}{\bsr}=\frac12\Big(-4r^2_1-r^2_3-4
r^2_7-3r^2_8+2r_1r_7-6r_3r_5-6r_2r_6-2\sqrt{3}r_3r_8+2\sqrt{3}r_5r_8\Big).
\end{eqnarray*}
Note
\begin{eqnarray}
\Inner{\bsa\star\bsa+\bsb\star\bsb}{\bsr}-2\Pa{\Inner{\bsa}{\bsr}^2+\Inner{\bsb}{\bsr}^2}
= -\Br{\sqrt{3}r_5+2(r_1-r_7)^2+\frac12(r_3+\sqrt{3}r_8)^2}
\end{eqnarray}
In fact, we can identify
$\ell_{\bsa,\bsb}=\min_{\bsr\in\Omega^{\mathrm{ext}}_3}\Innerm{\bsr}{\bsT_{\bsa,\bsb}}{\bsr}=-\frac{147}{64}$
by using the method from \cite{Szymanski2019}. In this situation, we
present the specific form of qutrit state saturating the lower bound
in the additive uncertainty relation
$\var_\rho(\bsA)+\var_\rho(\bsB)\geqslant\frac{15}{32}$, where such
inequality is saturated only when
$\rho(\epsilon)=\proj{\psi(\epsilon)}$, where
\begin{eqnarray}
\ket{\psi(\epsilon)}=\frac{\sqrt{14}}8\mathrm{i}\ket{0}+\epsilon\frac34\mathrm{i}\ket{1}+\frac{\sqrt{14}}8\ket{2},\quad\epsilon\in\set{\pm1}.
\end{eqnarray}
We can rewrite it as the generalized Bloch representation:
\begin{eqnarray}
\rho(\bsr) =
\frac13(\I+\sqrt{3}\bsr\cdot\bG)=\frac1{32}\Pa{\begin{array}{ccc}
                                                              7 & 3\sqrt{14}\epsilon & 7\mathrm{i} \\
                                                              3\sqrt{14}\epsilon & 18 & 3\sqrt{14}\mathrm{i}\epsilon \\
                                                              -7\mathrm{i} & -3\sqrt{14}\mathrm{i}\epsilon &
                                                              7
                                                            \end{array}
},
\end{eqnarray}
where
\begin{eqnarray}
\bsr=\Pa{\frac{6\sqrt{42}}{64}\epsilon,0,-\frac{11\sqrt{3}}{64},0,-\frac{14\sqrt{3}}{64},0,-\frac{6\sqrt{42}}{64}\epsilon,\frac{11}{64}}.
\end{eqnarray}
This example can be further generalized to the case where
$$
\bsA_t=\Pa{\begin{array}{ccc}
                                      -1 & 0 & t \\
                                      0 & 0 & 0 \\
                                      t & 0 & 1
                                    \end{array}
}\quad\text{and}\quad\bsB=\Pa{\begin{array}{ccc}
                   0 & 1 & 0 \\
                   1 & 0 & \mathrm{i} \\
                   0 & -\mathrm{i} & 0
                 \end{array}
}\quad (\forall t\in\real),
$$
then we have
$\min_{\rho\in\density{\bbC^3}}\Br{\var_\rho(\bsA_t)+\var_\rho(\bsB)}
=h(t)$, where
\begin{eqnarray}\label{eq:ht}
h(t):=
\begin{cases}
\frac{(15-t^2)(1+t^2)}{32},&\text{if }\abs{t}\leqslant1,\\
\frac{3+4t^2}{4(1+t^2)},&\text{if } \abs{t} \geqslant1.
\end{cases}
\end{eqnarray}
\end{exam}

\begin{exam}
For given two qutrit observables
$$
\bsA=\Pa{
\begin{array}{ccc}
 1 & 0 & 1 \\
 0 & -1 & -\mathrm{i} \\
 1 & \mathrm{i} & 0 \\
\end{array}
}\quad\text{and}\quad\bsB=\Pa{
\begin{array}{ccc}
 0 & 1 & -\mathrm{i} \\
 1 & 0 & 1 \\
 \mathrm{i} & 1 & 0 \\
\end{array}
},
$$
by using the method in \cite{Szymanski2019}, we have that
$$
\min_{\rho\in\rD(\bbC^3)}[\var_\rho(\bsA)+\var_\rho(\bsB)]=\min_{x\in\real}f(x),
$$
where
\begin{eqnarray*}
f(x)&:=& - \frac23(40x^2-52x+31)^{\frac12}\sin\Br{\frac\pi6 +
\frac13\arccos\Pa{-\frac{4x^3-186x^2+192x-23}{(40x^2-52x+31)^{\frac32}}}}\\
&&+\frac13(3x^2-2x+13).
\end{eqnarray*}
In fact, $\min_{x\in\real}f(x)\doteq0.427938$, which is attained at
$x\doteq-1.34253$. Note that the extremal qutrit state is
approximately given by
\begin{eqnarray*}
\rho = \Pa{
\begin{array}{ccc}
 0.14742387711339922 & 0.28637384163565377 \mathrm{i} & -0.20899784782884737 \\
 -0.28637384163565377\mathrm{i} & 0.5562869379027389 & 0.4059825161855749\mathrm{i} \\
 -0.20899784782884737 & -0.4059825161855749\mathrm{i} & 0.2962891849838617 \\
\end{array}
}.
\end{eqnarray*}
The above computation can be also checked by our algorithm presented
in Appendix~\ref{app:algorithm}.
\end{exam}

%=======================================================================%
\section{Discussions}\label{sect:6}
%=======================================================================%

%In section~\ref{sect:3}, we established the equivalence between the
%state-independent uncertainty relation in terms of the sum of
%variances and the quadratic programming with nonlinearly
%constraints. Based on this equivalence, we performed a lot of
%specific computations with special observables in both
%section~\ref{sect:4} and section~\ref{sect:5}. In fact, as toy
%models, we mainly focus on the derivations of tight lower bounds for
%the sum of variances of qubit/qutrit observables.
%
%Now we turn to consider the potential application of the additive
%uncertainty relation to entanglement detection. In fact, the
%entanglement detection by means of variance has been proposed in
%\cite{Guhne2004,Hofmann2003,Toth2004} previously. By the main
%results we derived here, we can provide a similar method to detect
%entanglement.

In section~\ref{sect:3}, we established the equivalence between the
state-independent uncertainty relation expressed as the sum of
variances and the quadratic programming problem with nonlinear
constraints. Leveraging this equivalence, we conducted detailed
computations involving specific observables in sections~\ref{sect:4}
and~\ref{sect:5}. Using qubit and qutrit systems as illustrative
examples, we derived tight lower bounds for the sum of variances of
observables.

Having explored the theoretical underpinnings of our uncertainty
relation, we now shift our focus to its practical implications. In
particular, we consider the potential application of the additive
uncertainty relation to entanglement detection. Previously, the use
of variance for entanglement detection has been proposed in
\cite{Guhne2004,Hofmann2003,Toth2004}. Drawing on our main results,
we can provide a similar method to detect entanglement, which offers
a novel approach to tackle this important problem in quantum
information theory. Specifically, let us suppose $\bsA_i$ are any
observables acting on $\complex^m$ and $\bsB_i$ are any observables
acting on $\complex^n$. First we consider the observables
$\bsM_i=\bsA_i\ot\I+\I\ot\bsB_i(i=1,2)$. For any \emph{separable}
bipartite state $\rho$ in $\density{\complex^m\ot\complex^n}$, we
have
\begin{eqnarray*}
\var_\rho(\bsM_1)+\var_\rho(\bsM_2)&\geqslant&
\min_\sigma\Br{\var_\sigma(\bsA_1)+\var_\sigma(\bsA_2)}+\min_\tau\Br{\var_\tau(\bsB_1)+\var_\tau(\bsB_2)}\\
&=&m(\bsA_1,\bsA_2)+m(\bsB_1,\bsB_2).
\end{eqnarray*}
So if some bipartite state $\varrho$ is such that
\begin{eqnarray}\label{eq:test1}
\var_\varrho(\bsM_1)+\var_\varrho(\bsM_2)<m(\bsA_1,\bsA_2)+m(\bsB_1,\bsB_2),
\end{eqnarray}
then $\varrho$ must be \emph{entangled}. The additive uncertainty
relation gives rise to an entanglement criterion \cite{Zhang2021} .
Second, for any bipartite state
$\omega\in\density{\complex^m\ot\complex^n}$, if it satisfies that
\begin{eqnarray}\label{eq:test2}
\var_{\omega^\Gamma}(\bsA)+\var_{\omega^\Gamma}(\bsB)<\min_{\rho\in\density{\complex^m\ot\complex^n}}\Br{\var_\rho(\bsA)+\var_\rho(\bsB)}=m(\bsA,\bsB)
\end{eqnarray}
for both bipartite observables $\bsA$ and $\bsB$, where
$\omega^\Gamma$ stands for the partial-transpose with respect to
either one subsystem, then $\omega$ must not be positive
partial-transposed state, i.e., it must be entangled state by PPT
criterion \cite{Peres1996}. This provides another way to detect
entanglement.

For further applications of our approach to entanglement detection,
we need to generate random quantum states to test the criteria
Eq.~\eqref{eq:test1} and Eq.~\eqref{eq:test2}. The theoretical
analysis becomes increasingly challenging due to the exponential
growth of parameters in quantum states with the expansion of the
underlying space's dimension, coupled with the absence of a
closed-form solution for the quadratic programming problem involving
nonlinear constraints in Remark~\ref{rem:3.3}. In our future
research, we intend to report the numerical results pertaining to
this aspect.

%============================================%
\section{Concluding remarks}\label{sect:7}
%============================================%

We conclude this paper by summarizing our findings and highlighting
future research directions:

Firstly, we have thoroughly investigated the additive uncertainty
relation of any two or more observables in qudit systems,
specifically focusing on variance using the generalized Gell-Mann
representation. To establish a tight state-independent lower bound,
we have constructed a universal optimization model aimed at
minimizing the sum of matrix variances. This model is firmly
grounded in constrained quadratic programming and the constrained
numerical range of real symmetric matrices, derived from the
generators of $\SU(n)$ and generalized Bloch vectors of observables.

Secondly, our exploration extended to lower-dimensional qubit and
qutrit systems. In the qubit realm, we derived an analytical lower
bound for the additive uncertainty relation applicable to any two or
more observables. For qutrit systems, we developed a general
algorithm to compute the analytical lower bound for pairs of
observables, complementing this with illustrative examples and
numerical computations.

Thirdly, we discussed the potential application of the additive
uncertainty relation in entanglement detection, emphasizing its
experimental feasibility. We anticipate that our results will
facilitate experimental implementations in this context.

Looking ahead, we identify two key problems for future research:
\begin{itemize}
\item Characterizing both endpoints of the constrained numerical range of
$\bsT$ mentioned in Remark~\ref{rem:3.3}. Specifically, determining
under what conditions these endpoints correspond to the minimal and
maximal eigenvalues, $\lambda_{\min}(\bsT)$ and
$\lambda_{\max}(\bsT)$, respectively, of $\bsT$. This problem
carries significant implications for matrix optimizations and could
pave the way for the development of more efficient quantum
algorithms.
\item For any bipartite entangled state $\varrho$, it remains an intriguing
challenge to find two observables $\bsA$ and $\bsB$ whose sum of
variances, $\var_{\varrho}(\bsA)+\var_{\varrho}(\bsB)$, can
nonlinearly witness entanglement in $\varrho$ in the manner
described previously. Addressing this challenge could lead to new
insights and techniques for entanglement detection in quantum
systems.
\end{itemize}

\subsubsection*{Acknowledgements}

This research was supported by Zhejiang Provincial Natural Science
Foundation of China under Grant No. LZ23A010005, and by the National
Natural Science Foundation of China (NSFC) under Grants No. 11971140
and No. 12171044. The first author would like to express his
gratitude to Prof. Jing-Ling Chen for his insightful comments on
this manuscript, and also to Prof. Hua Xiang for his improved
algorithm outlined in Appendix~\ref{app:algorithm}.

%--------------------------------------------------%

%--------------------------------------------------%
\newpage
\appendix
\appendixpage
\addappheadtotoc

%----------------------------------------------------------------------------------%
\section{Specific forms of all $\bsD_k$'s in qutrit system}\label{app:3Dk}
%----------------------------------------------------------------------------------%

In the following, each $\bsD_k(k=1,\ldots,8)$ mentioned in
Proposition~\ref{prop:starprod8} is an $8\times 8$ real symmetric
matrix that is given immediately.{\small
\begin{eqnarray*}
\bsD_1=\Pa{
\begin{array}{cccccccc}
 0 & 0 & 0 & 0 & 0 & 0 & 0 & 1 \\
 0 & 0 & 0 & 0 & 0 & 0 & 0 & 0 \\
 0 & 0 & 0 & 0 & 0 & 0 & 0 & 0 \\
 0 & 0 & 0 & 0 & 0 & \frac{\sqrt{3}}{2} & 0 & 0 \\
 0 & 0 & 0 & 0 & 0 & 0 & \frac{\sqrt{3}}{2} & 0 \\
 0 & 0 & 0 & \frac{\sqrt{3}}{2} & 0 & 0 & 0 & 0 \\
 0 & 0 & 0 & 0 & \frac{\sqrt{3}}{2} & 0 & 0 & 0 \\
 1 & 0 & 0 & 0 & 0 & 0 & 0 & 0 \\
\end{array}
}, \bsD_2=\Pa{
\begin{array}{cccccccc}
 0 & 0 & 0 & 0 & 0 & 0 & 0 & 0 \\
 0 & 0 & 0 & 0 & 0 & 0 & 0 & 1 \\
 0 & 0 & 0 & 0 & 0 & 0 & 0 & 0 \\
 0 & 0 & 0 & 0 & 0 & 0 & -\frac{\sqrt{3}}{2} & 0 \\
 0 & 0 & 0 & 0 & 0 & \frac{\sqrt{3}}{2} & 0 & 0 \\
 0 & 0 & 0 & 0 & \frac{\sqrt{3}}{2} & 0 & 0 & 0 \\
 0 & 0 & 0 & -\frac{\sqrt{3}}{2} & 0 & 0 & 0 & 0 \\
 0 & 1 & 0 & 0 & 0 & 0 & 0 & 0 \\
\end{array}
}
\end{eqnarray*}
\begin{eqnarray*}
\bsD_3=\Pa{
\begin{array}{cccccccc}
 0 & 0 & 0 & 0 & 0 & 0 & 0 & 0 \\
 0 & 0 & 0 & 0 & 0 & 0 & 0 & 0 \\
 0 & 0 & 0 & 0 & 0 & 0 & 0 & 1 \\
 0 & 0 & 0 & \frac{\sqrt{3}}{2} & 0 & 0 & 0 & 0 \\
 0 & 0 & 0 & 0 & \frac{\sqrt{3}}{2} & 0 & 0 & 0 \\
 0 & 0 & 0 & 0 & 0 & -\frac{\sqrt{3}}{2} & 0 & 0 \\
 0 & 0 & 0 & 0 & 0 & 0 & -\frac{\sqrt{3}}{2} & 0 \\
 0 & 0 & 1 & 0 & 0 & 0 & 0 & 0 \\
\end{array}
}, \bsD_4=\Pa{
\begin{array}{cccccccc}
 0 & 0 & 0 & 0 & 0 & \frac{\sqrt{3}}{2} & 0 & 0 \\
 0 & 0 & 0 & 0 & 0 & 0 & -\frac{\sqrt{3}}{2} & 0 \\
 0 & 0 & 0 & \frac{\sqrt{3}}{2} & 0 & 0 & 0 & 0 \\
 0 & 0 & \frac{\sqrt{3}}{2} & 0 & 0 & 0 & 0 & -\frac{1}{2} \\
 0 & 0 & 0 & 0 & 0 & 0 & 0 & 0 \\
 \frac{\sqrt{3}}{2} & 0 & 0 & 0 & 0 & 0 & 0 & 0 \\
 0 & -\frac{\sqrt{3}}{2} & 0 & 0 & 0 & 0 & 0 & 0 \\
 0 & 0 & 0 & -\frac{1}{2} & 0 & 0 & 0 & 0 \\
\end{array}
}
\end{eqnarray*}

\begin{eqnarray*}
\bsD_5=\Pa{
\begin{array}{cccccccc}
 0 & 0 & 0 & 0 & 0 & 0 & \frac{\sqrt{3}}{2} & 0 \\
 0 & 0 & 0 & 0 & 0 & \frac{\sqrt{3}}{2} & 0 & 0 \\
 0 & 0 & 0 & 0 & \frac{\sqrt{3}}{2} & 0 & 0 & 0 \\
 0 & 0 & 0 & 0 & 0 & 0 & 0 & 0 \\
 0 & 0 & \frac{\sqrt{3}}{2} & 0 & 0 & 0 & 0 & -\frac{1}{2} \\
 0 & \frac{\sqrt{3}}{2} & 0 & 0 & 0 & 0 & 0 & 0 \\
 \frac{\sqrt{3}}{2} & 0 & 0 & 0 & 0 & 0 & 0 & 0 \\
 0 & 0 & 0 & 0 & -\frac{1}{2} & 0 & 0 & 0 \\
\end{array}
}, \bsD_6=\Pa{
\begin{array}{cccccccc}
 0 & 0 & 0 & \frac{\sqrt{3}}{2} & 0 & 0 & 0 & 0 \\
 0 & 0 & 0 & 0 & \frac{\sqrt{3}}{2} & 0 & 0 & 0 \\
 0 & 0 & 0 & 0 & 0 & -\frac{\sqrt{3}}{2} & 0 & 0 \\
 \frac{\sqrt{3}}{2} & 0 & 0 & 0 & 0 & 0 & 0 & 0 \\
 0 & \frac{\sqrt{3}}{2} & 0 & 0 & 0 & 0 & 0 & 0 \\
 0 & 0 & -\frac{\sqrt{3}}{2} & 0 & 0 & 0 & 0 & -\frac{1}{2} \\
 0 & 0 & 0 & 0 & 0 & 0 & 0 & 0 \\
 0 & 0 & 0 & 0 & 0 & -\frac{1}{2} & 0 & 0 \\
\end{array}
}
\end{eqnarray*}
\begin{eqnarray*}
\bsD_7=\Pa{
\begin{array}{cccccccc}
 0 & 0 & 0 & 0 & \frac{\sqrt{3}}{2} & 0 & 0 & 0 \\
 0 & 0 & 0 & -\frac{\sqrt{3}}{2} & 0 & 0 & 0 & 0 \\
 0 & 0 & 0 & 0 & 0 & 0 & -\frac{\sqrt{3}}{2} & 0 \\
 0 & -\frac{\sqrt{3}}{2} & 0 & 0 & 0 & 0 & 0 & 0 \\
 \frac{\sqrt{3}}{2} & 0 & 0 & 0 & 0 & 0 & 0 & 0 \\
 0 & 0 & 0 & 0 & 0 & 0 & 0 & 0 \\
 0 & 0 & -\frac{\sqrt{3}}{2} & 0 & 0 & 0 & 0 & -\frac{1}{2} \\
 0 & 0 & 0 & 0 & 0 & 0 & -\frac{1}{2} & 0 \\
\end{array}
},\bsD_8=\Pa{\begin{array}{cccccccc}
      1 & 0 & 0 & 0 & 0 & 0 & 0 & 0 \\
      0 & 1 & 0 & 0 & 0 & 0 & 0 & 0 \\
      0 & 0 & 1 & 0 & 0 & 0 & 0 & 0 \\
      0 & 0 & 0 & -\frac{1}{2} & 0 & 0 & 0 & 0 \\
      0 & 0 & 0 & 0 & -\frac{1}{2} & 0 & 0 & 0 \\
      0 & 0 & 0 & 0 & 0 & -\frac{1}{2} & 0 & 0 \\
      0 & 0 & 0 & 0 & 0 & 0 & -\frac{1}{2} & 0 \\
      0 & 0 & 0 & 0 & 0 & 0 & 0 & -1
    \end{array}}.
\end{eqnarray*}}

\newpage

%--------------------------------------------------------------------%
\section{Algorithm for calculating
$\ell_{\bsa,\bsb}$ in Theorem~\ref{th:main}}\label{app:algorithm}
%--------------------------------------------------------------------%

\begin{algorithm}\caption{\label{algo:QP}
The minimum of $\Innerm{\bsr}{\bsT_{\bsa,\bsb}}{\bsr}$ where
$\bsr\in\Omega^{\mathrm{ext}}_3$} \KwIn{$\bsa$ and $\bsb$ in
$\real^8$} \KwOut{$\ell_{\bsa,\bsb}=\min_{\bsr\in
 \Omega^{\text{ext}}_3}\Innerm{\bsr}{\bsT_{\bsa,\bsb}}{\bsr}$ and $\bsr_{\min}= \arg\ell_{\bsa,\bsb}$}
\BlankLine\BlankLine
\begin{algorithmic}
   \STATE{ 01. Form $\bsT_{\bsa,\bsb}$ by using the matrices $\bsD_k(k=1,2,\cdots, 8)$, given in Appendix~\ref{app:3Dk}, and  $O_{\bsa,\bsb}=\proj{\bsa}+\proj{\bsb}$. }
   \STATE{ 02. Set $\bsr = (0, 0, 0, 0, 0, 0, 0, -1 )$ and compute $f_{\min} =\Innerm{\bsr}{\bsT_{\bsa,\bsb}}{\bsr}$, $\bsr_{\min} = \bsr$. }
   \STATE{ 03. Repeatedly sample three normally distributed random numbers and perform the following three steps: }
   \STATE{ 04. \qquad  Assign these numbers to $r_1, r_2, r_3$ respectively, such that $\sum_{j=1}^3 r_j^2 = \frac{3}{4}$; }
   \STATE{ 05. \qquad  Set $\bsr = (r_1, r_2, r_3, 0,0,0,0, \frac{1}{2}) $ and compute $f(\bsr)=\Innerm{\bsr}{\bsT_{\bsa,\bsb}}{\bsr}$; }
   \STATE{ 06. \qquad  If $f(\bsr) < f_{\min}$, then    $f_{\min} = f(\bsr)$, $\bsr_{\min} = \bsr$. }
   \STATE{ 07. For $k= 1,2,\cdots, N$, Do }
   \STATE{ 08. \qquad  Set $R= \frac{\sqrt{3}}{2N} k$, where $\frac{\sqrt{3}}{2}$ is equally divided into $N$ portions.  }
   \STATE{ 09. \qquad  Repeatedly sample four normally distributed random numbers and perform the following eight steps:}
   \STATE{ 10. \qquad\qquad   Assign these numbers to $r_4, r_5, r_6, r_7$ respectively, such that $\sum_{j=4}^7 r_j^2 = R^2$; }
   \STATE{ 11. \qquad\qquad   $s =    \frac{\sqrt{3}+\epsilon\sqrt{3-4R^2}}{2R^2}$, where $\epsilon\in\set{\pm1}$; }
   \STATE{ 12. \qquad\qquad   $r_1 = s(r_4 r_6 + r_5 r_7)$;}
   \STATE{ 13. \qquad\qquad   $r_2 = s(r_5 r_6 - r_4 r_7)$;}
   \STATE{ 14. \qquad\qquad   $r_3 = \frac{1}{2}s (r^2_4 + r^2_5-r^2_6- r^2_7)$;}
   \STATE{ 15. \qquad\qquad   $r_8 = \frac{\sqrt{3}}{2}s R^2 -1  $;}
   \STATE{ 16. \qquad\qquad   Set $\bsr = (r_1, \cdots, r_8)$ and compute $f(\bsr)=\Innerm{\bsr}{\bsT_{\bsa,\bsb}}{\bsr}$; }
   \STATE{ 17. \qquad\qquad   If $f(\bsr) < f_{\min}$, then    $f_{\min} = f(\bsr)$, $\bsr_{\min} = \bsr$. }
   \STATE{ 18. Endfor }
   \STATE{ 19. $\ell_{\bsa,\bsb} = f_{\min} $. }
\end{algorithmic}
\end{algorithm}

\end{document}

%% file: qudit_preamble.tex
\newcommand{\jpa}{J. Phys. A~}

\newcommand{\prl}{Phys. Rev. Lett.~}
\newcommand{\pra}{Phys. Rev. A~}
\newcommand{\pla}{Phys. Lett. A~}
\newcommand{\rmp}{Rev. Math. Phys.~}

\usepackage{algorithmic}
\usepackage{young}
\usepackage[latin1]{inputenc}
\usepackage[linesnumbered,boxed,ruled,commentsnumbered]{algorithm2e}
\usepackage{appendix}
\usepackage{epstopdf}
\usepackage{xcolor}
\usepackage{subfigure}
\usepackage[T1]{fontenc}
\usepackage[sc]{mathpazo}
\usepackage{amsmath}
\usepackage{amssymb}
\usepackage{graphicx,color}
\usepackage{framed}
\usepackage{multirow}
\usepackage{enumerate}
\usepackage{amsthm}% theorems and proofs
\usepackage{amsfonts,mathrsfs}
\usepackage{geometry} % allows easy specifying of the page layout
\usepackage{eepic}
\usepackage{ifthen}
\usepackage[vcentermath]{youngtab}
\usepackage[unicode=true,pdfusetitle, bookmarks=true,bookmarksnumbered=false,bookmarksopen=false, breaklinks=false,pdfborder={0 0 0},backref=false,colorlinks=false] {hyperref}
\hypersetup{
colorlinks,linkcolor=myurlcolor,citecolor=myurlcolor,urlcolor=myurlcolor}
\definecolor{myurlcolor}{rgb}{0,0,0.7}% Color highlighted for numbers in citations and theorems

\usepackage{color}

\newcommand{\blue}{\textcolor{blue}}

\newcommand{\proj}[1]{| #1\rangle\!\langle #1 |}

\usepackage{hyperref}
\hypersetup{pdfpagemode=UseNone}

\newcommand{\tinyspace}{\mspace{1mu}}

\newcommand{\op}[1]{\operatorname{#1}}

\newcommand{\abs}[1]{\left\lvert\tinyspace #1 \tinyspace\right\rvert}

\newcommand{\norm}[1]{\left\lVert\tinyspace #1 \tinyspace\right\rVert}

\renewcommand{\det}{\operatorname{det}}

\newcommand{\setft}[1]{\mathrm{#1}}

\newcommand{\density}[1]{\setft{D}\left(#1\right)}

\renewcommand{\vec}{\op{vec}}
\newcommand{\im}{\op{im}}

 \newcommand{\var}{\operatorname{Var}}

\def\SU{\mathsf{SU}}

\def\su{\mathfrak{su}}

\def \diag {\mathrm{diag}}

\def \re {\mathrm{Re}}
\def \im {\mathrm{Im}}

\def\complex{\mathbb{C}}
\def\real{\mathbb{R}}

\def\I{\mathbb{1}}

\def\zero{\mathbf{0}}

\newenvironment{mylist}[1]{\begin{list}{}{
    \setlength{\leftmargin}{#1}
    \setlength{\rightmargin}{0mm}
    \setlength{\labelsep}{2mm}
    \setlength{\labelwidth}{8mm}
    \setlength{\itemsep}{0mm}}}
    {\end{list}}

\def\ot{\otimes}

\newcommand{\Inner}[2]{\left\langle #1 , #2\right\rangle}
\newcommand{\Innerm}[3]{\left\langle #1 \left| #2 \right| #3 \right\rangle}
\newcommand{\defeq}{\stackrel{\smash{\textnormal{\tiny def}}}{=}}

\newcommand{\Pa}[1]{\left(#1\right)}

\newcommand{\Br}[1]{\left[#1\right]}
\newcommand{\set}[1]{\{#1\}}
\newcommand{\Set}[1]{\left\{#1\right\}}

\newcommand{\ket}[1]{|#1\rangle}

\DeclareMathOperator{\trace}{Tr}

\newcommand{\Ptr}[2]{\trace_{#1}\Pa{#2}}

\newcommand{\Tr}[1]{\Ptr{}{#1}}

\def\cD{\mathcal{D}}
\def\cI{\mathcal{I}}\def\cJ{\mathcal{J}}

\def\cU{\mathcal{U}}

\def\bG{\mathbf{G}}

\def\bsA{\boldsymbol{A}}\def\bsB{\boldsymbol{B}}\def\bsD{\boldsymbol{D}}
\def\bsG{\boldsymbol{G}}
\def\bsL{\boldsymbol{L}}\def\bsM{\boldsymbol{M}}\def\bsO{\boldsymbol{O}}
\def\bsT{\boldsymbol{T}}
\def\bsX{\boldsymbol{X}}\def\bsY{\boldsymbol{Y}}

\def\bsa{\boldsymbol{a}}\def\bsb{\boldsymbol{b}}

\def\bsr{\boldsymbol{r}}
\def\bsx{\boldsymbol{x}}\def\bsy{\boldsymbol{y}}

\def\rD{\mathrm{D}}

\def\L{\textsf{L}}

\def\bbC{\mathbb{C}}

\def\bbR{\mathbb{R}}

\newtheorem{thrm}{Theorem}[section]
\newtheorem{lem}[thrm]{Lemma}
\newtheorem{prop}[thrm]{Proposition}
\newtheorem{cor}[thrm]{Corollary}
\theoremstyle{definition}
\newtheorem{definition}[thrm]{Definition}
\newtheorem{remark}[thrm]{Remark}
\newtheorem{exam}[thrm]{Example}

\numberwithin{equation}{section}

\newcounter{questionnumber}

%% file: 2024physscr.bbl
\begin{thebibliography}{999}

\bibitem{Heisenberg1927}
W. Heisenberg, \"{U}ber den anschaulichen Inhalt der
quantentheoretischen Kinematik und Mechanik, Z. Phys.
\href{https://doi.org/10.1007/BF01397280}{{\bf 43}, 172-198 (1927).}

\bibitem{Robertson1929}
H.P. Robertson, The Uncertainty Principle, Phys. Rev.
\href{https://doi.org/10.1103/PhysRev.34.163}{{\bf 34}, 163 (1929).}

\bibitem{Friedland2013}
S. Friedland, V. Gheorghiu, G. Gour, Universal Uncertainty
Relations, \prl
\href{https://doi.org/10.1103/PhysRevLett.111.230401}{{\bf 111},
230401 (2013).}

\bibitem{Pucha2013}
Z. Pucha{\l}a, {\L}. Rudnicki, K. \.{Z}yczkowski, Majorization
entropic uncertainty relations, \jpa: Math. Theor.
\href{https://doi.org/10.1088/1751-8113/46/27/272002}{{\bf 46},
272002 (2013).}

\bibitem{Deutsch1983}
D. Deutsch, Uncertainty in Quantum Measurements, \prl
\href{https://doi.org/10.1103/PhysRevLett.50.631}{{\bf 50}, 631
(1983).}

\bibitem{Wu2009}
S. Wu, S. Yu, K. M{\o}lmer, Entropic uncertainty relation for
mutually unbiased bases, \pra
\href{https://doi.org/10.1103/PhysRevA.79.022104}{{\bf 79}, 022104
(2009).}

\bibitem{Coles2017}
P.J. Coles, M. Berta, M. Tomamichel, S. Wehner, Entropic uncertainty
relations and their applications, \rmp
\href{https://doi.org/10.1103/RevModPhys.89.015002}{{\bf 89}, 015002
(2017).}

\bibitem{Maassen1988}
H. Maassen and J.B.M. Uffink, Generalized Entropic Uncertainty
Relations, \prl
\href{https://doi.org/10.1103/PhysRevLett.60.1103}{{\bf 60}, 1103
(1988).}

\bibitem{Renes2009}
J.M. Renes, J.C. Boileau, Conjectured strong complementary
information tradeoff, \prl
\href{https://doi.org/10.1103/PhysRevLett.103.020402}{{\bf 103},
020402 (2009).}

\bibitem{Gour2018}
G. Gour, A. Grudka, M. Horodecki, W. K{\l}obus, J. {\L}odyga, V.
Narasimhachar, Conditional uncertainty principle, \pra
\href{https://doi.org/10.1103/PhysRevA.97.042130}{{\bf 97}, 042130
(2018).}

\bibitem{Kurzyk2018}
D. Kurzyk, {\L}. Pawela, Z. Pucha{\l}a, Conditional entropic
uncertainty relations for Tsallis entropies, Quantum Inf. Process.
\href{https://doi.org/10.1007/s11128-018-1955-1}{{\bf 17}, 193
(2018).}

\bibitem{Grudka2013}
A. Grudka, M. Horodecki, P. Horodecki, R. Horodecki, W. K{\l}obus,
{\L}. Pankowski, Conjectured strong complementary-correlations
tradeoff, \pra
\href{https://doi.org/10.1103/PhysRevA.88.032106}{{\bf 88}, 032106
(2013).}

\bibitem{Tomamichel2011}
M. Tomamichel, R. Renner, Uncertainty Relation for Smooth Entropies,
\prl \href{https://doi.org/10.1103/PhysRevLett.106.110506}{{\bf
106}, 110506 (2011).}

\bibitem{Vallone2014}
G. Vallone, D.G. Marangon, M. Tomasin, P. Villoresi, Quantum
randomness certified by the uncertainty principle, \pra
\href{https://doi.org/10.1103/PhysRevA.90.052327}{{\bf 90}, 052327
(2014).}

\bibitem{Cao2016}
Z. Cao, H. Zhou, X. Yuan, X. Ma,  Source-Independent Quantum Random
Number Generation, Phys. Rev. X
\href{https://doi.org/10.1103/PhysRevX.6.011020}{{\bf 6}, 011020
(2016).}

\bibitem{Berta2014}
M. Berta, P.J. Coles, S. Wehner, Entanglement-assisted guessing of
complementary measurement outcomes, \pra
\href{https://doi.org/10.1103/PhysRevA.90.062127}{{\bf 90}, 062127
(2014).}

\bibitem{Walborn2011}
S.P. Walborn, A. Salles, R.M. Gomes, F. Toscano, P.H.S. Ribeiro,
Revealing Hidden Einstein-Podolsky-Rosen Nonlocality, \prl
\href{https://doi.org/10.1103/PhysRevLett.106.130402}{{\bf 106},
130402 (2011).}

\bibitem{Schneeloch2013}
J. Schneeloch, C.J. Broadbent, S.P. Walborn, E.G. Cavalcanti, J.C.
Howell, Einstein-Podolsky-Rosen steering inequalities from entropic
uncertainty relations, \pra
\href{https://doi.org/10.1103/PhysRevA.87.062103}{{\bf 87}, 062103
(2013).}

\bibitem{Giovannetti2011}
V. Giovannetti, S. Lloyd, L. Maccone, Advances in quantum metrology,
Nat. Photon. \href{https://doi.org/10.1038/nphoton.2011.35}{{\bf 5},
222 (2011).}

\bibitem{Zhang2021}
L. Zhang, S. Luo, S-M. Fei, and J. Wu, Uncertainty regions of
observables and state-independent uncertainty relations, Quant. Inf
Process \href{https://doi.org/10.1007/s11128-021-03303-w}{{\bf20}:
357 (2021).}

\bibitem{Szymanski2019}
K. Szyma\'{n}ski and Karol \.{Z}yczkowski, Geometric and algebraic
origins of additive uncertainty relations, \jpa: Math. Theor.
\href{https://doi.org/10.1088/1751-8121/ab4543}{{\bf53}, 015302
(2019).}

\bibitem{Schwonnek2017}
R. Schowonnek, L. Dammeier, and R.F. Werner, State-Independent
Ucertainty Relations and Entanglement Detection in Noisy Systems,
\prl
\href{https://doi.org/10.1103/PhysRevLett.119.1704046}{{\bf119},
170404 (2017).}

\bibitem{Zhao2019}
Y-Y. Zhao, G-Y. Xiang, X-M. Hu, B-H. Liu, C-F. Li, G-C. Guo, R.
Schwonnek, and R. Wolf, Entanglement Detection by Violations of
Noisy Uncertainty Relations: A Proof of Principle, \prl
\href{https://doi.org/10.1103/PhysRevLett.122.220401}{{\bf122},
220401 (2019).}

\bibitem{Guhne2004}
O. G\"{u}hne, Characterizing entanglement via uncertainty relations,
\prl \href{https://doi.org/10.1103/PhysRevLett.92.117903}{{\bf92},
117903 (2004).}

\bibitem{Akbari2018}
Y. Akbari-Kourbolagh and M. Azhdargalam, Entanglement criterion for
tripartite systems based on local sum uncertainty relations, \pra
\href{https://doi.org/10.1103/PhysRevA.97.042333}{{\bf97}, 042333
(2018).}

\bibitem{Maccone2014}
L. Maccone and A.K. Pati, Stronger Uncertainty Relations for All
Incompatible Observables, \prl
\href{https://doi.org/10.1103/PhysRevLett.113.260401}{{\bf113},
260401 (2014)}; Erratum
\href{https://doi.org/10.1103/PhysRevLett.114.039902}{{\bf114},
039902 (2015)}

\bibitem{Li2015}
J-L. Li, C-F. Qiao, Reformulating the quantum uncertainty relation,
Sci. Rep. \href{https://doi.org/10.1038/srep12708}{{\bf5}, 12708
(2015).}

\bibitem{Qian2018}
C. Qian, J-L. Li, and C-F. Qiao, State-independent uncertainty
relations and entanglement detection, Quant Inf Process
\href{https://doi.org/10.1007/s11128-018-1855-4}{{\bf17}, 84
(2018).}

\bibitem{Xiao2019}
Y. Xiao C. Guo, F. Meng, , N. Jing, and M-H. Yung, Incompatibility
of observables as state-independent bound of uncertainty relations,
\pra \href{https://doi.org/10.1103/PhysRevA.100.032118}{{\bf100},
032118 (2019).}

\bibitem{Byrd2003}
M.S. Byrd and N. Khaneja, Characterization of the positivity of the
density matrices in terms of the coherence vector representation,
\pra \href{https://doi.org/10.1103/PhysRevA.68.062322}{{\bf68},
062322 (2003).}

\bibitem{Kimura2003}
G. Kimura, The Bloch vector for $N$-level systems, \pla
\href{https://doi.org/10.1016/S0375-9601(03)00941-1}{{\bf314},
339-349 (2003).}

\bibitem{Loubenets2021}
E.R. Loubenets and M. Kulakov, The Bloch vectors formalism for a
finite-dimensional quantum system, J. Phys. A: Math. Theor.
\href{https://doi.org/10.1088/1751-8121/abf1ae}{{\bf54}, 195301
(2021).}

\bibitem{Nocedal2006}
J. Nocedal and S.J. Wright, Numerical optimization (2nd Ed),
Springer (2006).

\bibitem{Best2017}
M.J. Best, Quadratic Programming With Computer Programs, CRC Press
(2017).

\bibitem{Dostal2009}
Z. Dost\'{a}l, Optimal Quadratic Programming Algorithms, Springer
(2009).

\bibitem{Arvind1997}
Arvind, K.S. Mallesh and N. Mukunda, A generalized Pancharatnam
geometric phase formula for three-level quantum systems, \jpa: Math.
Gen. \href{https://doi.org/10.1088/0305-4470/30/7/021}{{\bf30}, 2417
(1997).}

\bibitem{Ercolessi2001}
E. Ercolessi, G. Marmo, and G. Morandi, Geometry of mixed states and
degeneracy structure of geometric phases for multi-level quantum
systems. A unitary group approach, Int.J.Mod.Phys.
\href{https://doi.org/10.1142/S0217751X01005870}{{\bf 16}(31), 5007
(2001).}

\bibitem{Goyal2016}
S.K. Goyal, B.N. Simon, R. Singh, and S. Simon, Geometry of the
generalized Bloch sphere for qutrits, \jpa: Math. Theor.
\href{https://doi.org/10.1088/1751-8113/49/16/165203}{{\bf49},
165203 (2016).}

\bibitem{Kurzynski2016}
P. Kurzy\'{n}ski, A. Ko{\l}odziejski, W. Laskowski, and M.
Markiewicz, Three-dimensional visualization of a qutrit, \pra
\href{http://dx.doi.org/10.1103/PhysRevA.93.062126}{{\bf93}, 062126
(2016).}

\bibitem{Gao2021}
L-M. Zhang, T. Gao, and L. Yan, Tighter uncertainty relations based
on Wigner-Yanase skew information for observables and channels, \pla
\href{https://doi.org/10.1016/j.physleta.2020.127029}{{\bf387},
127029 (2021).}

\bibitem{Toth2022}
G. T\'{o}th and F. Fr\"{o}wis, Uncertainty relations with the
variance and the quantum Fisher information based on convex
decompositions of density matrices, Phys. Rev. Research
\href{https://doi.org/10.1103/PhysRevResearch.4.013075}{{\bf4},
013075 (2022).}

\bibitem{Chiew2022}
S-H. Chiew and M. Gessner, Improving sum uncertainty relations with
the quantum Fisher information, Phys. Rev. Research
\href{https://doi.org/10.1103/PhysRevResearch.4.013076}{{\bf4},
013076 (2022).}

\bibitem{Hofmann2003}
H.F. Hofmann, S. Takeuchi, Violation of local uncertainty relations
as a signature of entanglement, \pra
\href{https://doi.org/10.1103/PhysRevA.68.032103}{{\bf68}, 032103
(2003).}

\bibitem{Toth2004}
G. T\'{o}th, Entanglement detection in optical lattices of bosonic
atoms with collective measurements, \pra
\href{https://doi.org/10.1103/PhysRevA.69.052327}{{\bf69}, 052327
(2004).}

\bibitem{Peres1996}
A. Peres, Separability criterion for density matrices, \prl
\href{https://doi.org/10.1103/PhysRevLett.77.1413}{{\bf77}, 1413
(1996).}

\end{thebibliography}
